\documentclass[reprint,superscriptaddress,aps,pre]{revtex4-2}
\usepackage{graphicx,amsmath,amssymb,mathtools,amsthm,bm,physics,algorithm,algpseudocode,xcolor,hyperref}
\DeclareMathOperator*{\argmin}{arg\,min}
\DeclareMathOperator{\diag}{diag}
\newcommand{\conjtp}[1]{{#1^{\dagger}}}
\newcommand{\defeq}{\coloneqq}

\newcommand{\ie}{{\textit{i.e.}}}

\newcommand{\svec}{{\hat{\mathbf{S}}}}
\newcommand{\sx}{{\hat{S}^{x}}}
\newcommand{\sy}{{\hat{S}^{y}}}
\newcommand{\sz}{{\hat{S}^{z}}}
\newcommand{\splus}{{\hat{S}^{+}}}
\newcommand{\sminus}{{\hat{S}^{-}}}
\newcommand{\ham}{{\hat{H}}}
\newcommand{\twist}{{\hat{U}_{\text{twist}}}}

\newcommand{\zmeansq}{{\overline{|z|^2}}}

\theoremstyle{definition}
\newtheorem{theorem}{Theorem}
\newtheorem{lemma}[theorem]{Lemma}

\begin{document}

\title{Nested Iterative Shift-invert Diagonalization for Many-body Localization in the Random-field Heisenberg Chain}

\author{Taito Kutsuzawa}
\affiliation{Department of Physics, The University of Tokyo, Tokyo 113-0033, Japan}
\author{Synge Todo}
\affiliation{Department of Physics, The University of Tokyo, Tokyo 113-0033, Japan}
\affiliation{Institute for Physics of Intelligence, The University of Tokyo, Tokyo 113-0033, Japan}
\affiliation{Institute for Solid State Physics, The University of Tokyo, Kashiwa 277-8581, Japan}

\date{\today}

\begin{abstract}
We study the many-body localization of the random-field Heisenberg chain using the nested shift-invert Lanczos method with an iterative linear solver.
We use the minimum residual method (MINRES) inside each Lanczos iteration.
The memory consumption of the proposed method is only proportional to the dimension of the Hilbert space.
We also introduce a preconditioner that takes into account the effects of disorder and interaction in the random-field Heisenberg chain.
As a probe of many-body localization transition, we propose a unitary operator called the twist operator, which has a clear interpretation in the real space.
We discuss its behavior for thermal and localized eigenstates.
We demonstrate the efficiency of the nested iterative shift-invert diagonalization method with the proposed preconditioner for the many-body localization problem and estimate the transition point of the random-field Heisenberg chain more precisely based on the finite-size analysis of the expectation value of the twist operator.
\end{abstract}

\maketitle

\section{Introduction}
Understanding the mechanism of thermalization is one of the central problems in condensed matter physics~\cite{Alessio2016}.
A natural question is whether a closed quantum system thermalizes or not by its own unitary dynamics.
The eigenstate thermalization hypothesis (ETH)~\cite{Deutsch1991,Srednicki1994,Srednicki1999} is a sufficient condition that explains thermalization from the microscopic level.
However, it is known that there are several counterexamples, such as integrable systems~\cite{Rigol2008,Rigol2009,Santos2010,Biroli2010,Steinigeweg2013} and quantum many-body scar~\cite{Serbyn2021,Moudgalya2021,Papic2021}, where the system does not thermalize by itself.
Many-body localization~\cite{Nandkishore2015,Abanin2017,Alet2018,Abanin2019} has been investigated extensively as another possible mechanism that breaks ETH.
A many-body localized system fails to thermalize under strong disorder but can also delocalize and recover thermalization due to interaction.
The phase transition between thermal and localized phases is called many-body localization transition~\cite{Pal2010,Luitz2015}.

The Hamiltonian of the spin-1/2 random-field Heisenberg chain is given by
\begin{equation} \label{eq:heisenberg_hamiltonian}
    \ham = \sum_{j = 1}^{L} (\svec_{j} \cdot \svec_{j + 1} - h_j \sz_{j}),
\end{equation}
where $\svec_{j} = (\sx_{j}, \sy_{j}, \sz_{j})$ is a spin-1/2 operator at site $j = 1, 2, \dots, L$.
We impose the periodic boundary conditions and identify $\svec_{L + 1}$ with $\svec_{1}$.
The Hamiltonian \eqref{eq:heisenberg_hamiltonian} commutes with $\sz \coloneqq \sum_{j=1}^L \sz_{j}$.
We consider a chain of even length $L$ and eigenstates in the $\sz = 0$ sector.
The magnetic fields $h_1, h_2, \dots, h_L$ are independent and identically distributed random variables with the uniform distribution over the interval $[-h, h]$, where $h \ge 0$ represents the disorder strength.

Let us consider highly excited states in the middle of the energy spectrum of the model \eqref{eq:heisenberg_hamiltonian}.
It has been argued that they behave qualitatively differently in the weak disorder regime $h \ll J$ and the strong disorder regime $h \gg J$~\cite{Pal2010,Luitz2015}.
Under the weak disorder, the energy eigenstates are extended and have volume-law entanglement (thermal phase).
On the other hand, under the strong disorder, the energy eigenstates localize and obey the area law of entanglement (many-body localized phase), where local integrals of motion emerge~\cite{Serbyn2013a} but the interaction causes dephasing dynamics that yields the logarithmic growth of entanglement~\cite{Znidaric2008,Bardarson2012,Serbyn2013b,Huse2014}.

The many-body localization has been investigated by using various numerical methods.
As long as the disorder is strong enough $h \gg J$, the tensor network method~\cite{Khemani2016,Xiongjie2017} can simulate many-body localized states efficiently as the entanglement entropy is small.
However, as this method becomes rapidly ineffective near the transition between volume-law and area-law entangled eigenstates,
the exact diagonalization has been used mainly so far~\cite{Pal2010,Luitz2015,Sierant2020,Beeumen2020}. 
Unfortunately, the system size that the exact diagonalization can handle is limited, and the numerical results suffer from a severe finite-size effect. Thus, the nature of the many-body localization transition in the thermodynamic limit is still under active study~\cite{Roeck2017,Luitz2017,Vosk2015,Potter2015,Dumitrescu2017,Thiery2018,Zhang2016,Goremykina2019,Morningstar2019,Morningstar2020,Dumitrescu2019}.
To assess the existence of the many-body localized phase in the infinite system, we must study larger systems with the help of more sophisticated numerical methods.

In this paper, we introduce the shift-invert Lanczos method~\cite{Ericsson1980,TemplatesEigen} with an iterative linear solver, as with the Anderson localization problem~\cite{Schenk2008}.
The most computationally demanding part of the algorithm is to solve a linear equation.
The coefficient matrix is the energy-shifted Hamiltonian, which is Hermitian indefinite.
For this purpose, we use the minimum residual method (MINRES)~\cite{Paige1975}, which is one of the Krylov subspace methods~\cite{MatrixComputations}.
The memory usage becomes lower than direct methods~\cite{Luitz2015,Pietracaprina2018} by a factor of the exponential of system size.

The computation time becomes longer than the direct method in exchange for lower memory usage.
For Krylov subspace methods to converge successfully, preconditioning is essential.
We propose a simple preconditioner with both disorder and interaction taken into account, which we find effective for the present random-field Heisenberg chain.

By using the proposed nested iterative method, we successfully reproduce the results of the bipartite entanglement entropy and the level spacing ratio~\cite{Oganesyan2007,Pal2010,Luitz2015} across the many-body localization transition of the random-field Heisenberg chain~\eqref{eq:heisenberg_hamiltonian}.
In addition, we propose to use the twist operator~\cite{Lieb1961,Nakamura2002} as a probe of many-body localization transition.
The twist operator has a clear meaning in the real space that it rotates spins over the chain with gradually increasing angles.
We discuss the finite-size effect of the expectation value of the twist operator for thermal and localized eigenstates.
We calculate the expectation value of the twist operator for each energy eigenstate to observe how thermal and localized eigenstates respond to the global perturbation induced by the twist operator.
Furthermore, we estimate the many-body localization transition point of the random-field Heisenberg chain by the finite-size analysis of the twist operator.

The rest of this paper is organized as follows: In Sec.~\ref{sec:nested_iterative_shift_invert_diagonalization}, we introduce the nested iterative shift-invert Lanczos method with MINRES.
We also propose a preconditioner that is effective to the random-field Heisenberg chain.
In Sec.~\ref{sec:twist_operator}, the twist operator is introduced as a probe of many-body localization transition.
We discuss the asymptotic behavior of its expectation value in thermal and localized phases.
In Sec.~\ref{sec:results}, we show the benchmark results of the present algorithm and give a finite-size analysis of the many-body localization transition by using the twist operator.
Finally, we give a conclusion in Sec.\ref{sec:conclusion}.

\section{Nested Iterative Shift-invert Diagonalization}
\label{sec:nested_iterative_shift_invert_diagonalization}
Let $A \in \mathbb{C}^{n \times n}$ be a Hermitian sparse matrix.
Our goal is to calculate a few eigenvectors of $A$ with eigenvalues closest to an arbitrary target value $\sigma$, which is typically in the middle of the spectrum of $A$.

We use the shift-invert Lanczos method~\cite{Ericsson1980,TemplatesEigen}.
In other words, we apply the Lanczos method~\cite{Kawamura2017} on the matrix $(A - \sigma I)^{-1}$, where $I \in \mathbb{C}^{n \times n}$ is the identity matrix.
The most time-consuming part of the algorithm is the matrix inversion:
\begin{equation}
    \bm{v} \mapsto (A - \sigma I)^{-1} \bm{v}
\end{equation}
for a given vector $\bm{v} \in \mathbb{C}^n$.
To execute the matrix inversion, a previous study~\cite{Luitz2015} used direct solvers, which compute the LU decomposition of the coefficient matrix $A - \sigma I$.

Instead, we make use of the minimum residual method (MINRES)~\cite{Paige1975}, which is a Krylov subspace method that solves a linear equation $A \bm{x} = \bm{b}$ with a Hermitian coefficient matrix $A \in \mathbb{C}^{n \times n}$ by minimizing the norm of the residual vector $\bm{b} - A \bm{x}$ over the Krylov subspace (see appendix~\ref{appendix:minres}).
Algorithm~\ref{alg:SIMINRES} shows the pseudocode of the nested iterative shift-invert diagonalization, where $\bm{v} \in \mathbb{C}^n$ is an initial vector and $m > 0$ the maximum number of Lanczos iterations.
The superscript $\dagger$ means Hermitian conjugate.
The iterative algorithm MINRES is nested inside each iteration of the Lanczos method.
\begin{figure}[tb]
  \begin{algorithm}[H]
      \caption{\label{alg:SIMINRES}Nested iterative shift-invert diagonalization}
      \begin{algorithmic}[1]
          \Function{SIMINRES}{$A,\ \sigma,\ \bm{v},\ m$}
            \State $\bm{v}_0 \gets \bm{0}$
            \State $\beta_1 \gets \norm{\bm{v}}$
            \State $\bm{v}_1 \gets \bm{v} / \beta_1$
            \State $k \gets 1$
            \Loop
              \State /* Solve $(A - \sigma I) \bm{x} = \bm{v}_k$ by MINRES */
              \State $\bm{x} \gets (A - \sigma I) ^{-1} \bm{v} _{k}$
              \State
              \State /* Orthogonalization */
              \State $\bm{w} \gets \bm{x} - \beta_k \bm{v}_{k - 1}$
              \State $\alpha_k \gets \conjtp{\bm{v}_k} \bm{w}$
              \State $\bm{w} \gets \bm{w} - \alpha_k \bm{v}_k$
              \State
              \State /* Normalization */
              \State $\beta_{k + 1} \gets \sqrt{\conjtp{\bm{w}} \bm{w}}$
              \If{$\beta_{k + 1} = 0$ or $k = m$}
                \State break
              \EndIf
              \State $\bm{v}_{k + 1} \gets \bm{w} / \beta_{k + 1}$
              \State
              \State $k \gets k + 1$
            \EndLoop
            \State Let $T_k$ be the $k \times k$ real symmetric tridiagonal matrix with diagonals $\alpha_1, \dots, \alpha_k$ and subdiagonals $\beta_2, \dots, \beta_k$.
            \State Diagonalize $T_k$ ($T_k \bm{y}_i^{(k)} = \theta_i^{(k)} \bm{y}_i^{(k)}\ (i = 1, 2, \dots, k)$).
            \State \Return $\sigma + \frac{1}{\theta_i^{(k)}}, V_k \bm{y}_i^{(k)}\ (i = 1, 2, \dots, k)$
            \EndFunction
      \end{algorithmic}
  \end{algorithm}
\end{figure}

We substitute the iterative solver for direct solvers because the former has a great advantage in memory usage against the latter.
In direct algorithms, as the coefficient matrix is factorized into lower and upper triangular matrices, some zero entries in the original matrix become nonzero.
These new nonzero elements are called \emph{fill-in}.
Let us define the fill-in ratio as the number of nonzero entries of the factorized matrices divided by that of the original matrix.
Direct solvers attempt to keep the fill-in ratio as small as possible.
However, it is reported~\cite{Pietracaprina2018} that the fill-in ratio grows exponentially with the system size for the random-field Heisenberg model's Hamiltonian \eqref{eq:heisenberg_hamiltonian} represented in the $\sz$ basis, where all $\sz_1, \sz_2, \dots, \sz_L$ are diagonal.
In contrast to direct solvers, iterative algorithms do not modify the coefficient matrix $A$ but use it only as matrix-vector products $\bm{v} \mapsto A \bm{v}$.
We do not even need to retain the Hamiltonian on memory if the matrix elements can be generated on the fly during the matrix-vector product.
Moreover, Krylov subspace methods for Hermitian coefficient matrices, including MINRES, usually retain only a fixed number of vectors as a result of the Lanczos tridiagonalization.
Thus, the memory usage can be kept down to $\mathcal{O}(n)$.

On the other hand, there is a possibility that iterative algorithms do not give solutions in a reasonable time or even do not converge at all.
The convergence of the Krylov subspace methods is governed by the condition numbers of coefficient matrices.
If the coefficient matrix is ill-conditioned, that is, its condition number is large, it takes many iterations for iterative algorithms to converge.
In fact, our problem becomes ill-conditioned exponentially with the system size because the target energy $\sigma$ is in the middle of the energy spectrum where the density of states grows exponentially.
However, the convergence of iterative algorithms is improved significantly if we can find an effective preconditioner.
In order to make the algorithm practical, it is essential to develop a preconditioner, tailored to our coefficient matrix, which transforms the matrix into one as close to the identity matrix as possible.

Let us consider solving a preconditioned system
\begin{equation}
  (C^{-1} (A - \sigma I) \conjtp{(C^{-1})})(\conjtp{C} \bm{x}) = C^{-1} \bm{v}
\end{equation}
by MINRES.
The preconditioning matrix is defined to be $M \coloneqq C \conjtp{C}$, which must be easily invertible and positive definite.
We consider diagonal matrices with positive diagonal elements to satisfy these requirements.
Let $A$ be the Hamiltonian \eqref{eq:heisenberg_hamiltonian} represented in the $\sz$ basis.
If the matrix $A - \sigma I = (a_{ij} - \sigma \delta_{ij})_{1 \le i, j \le n}$ is diagonally dominant, then scaling the matrix by its diagonal elements should be effective.
This idea gives rise to the Jacobi preconditioner~\cite{MatrixComputations}, namely
\begin{equation}
  M_{\text{Jacobi}} \defeq \diag (|a_{11} - \sigma|, |a_{22} - \sigma|, \dots, |a_{nn} - \sigma|),
\end{equation}
where we assume $a_{ii} \neq \sigma\ (i = 1, 2, \dots, n)$ and take the absolute values to ensure the positive definiteness.

The Hamiltonian \eqref{eq:heisenberg_hamiltonian} does become diagonally dominant when the disorder is strong.
However, the off-diagonal elements play an essential role in the delocalization transition to the thermal phase as the disorder is weakened.
For this reason, we attempt to incorporate off-diagonal elements into the Jacobi preconditioner by considering the following preconditioner:
\begin{equation} \label{eq:rownorm_preconditioner}
  M_{\text{norm}} \defeq \diag (r_1, r_2, \dots, r_n),
\end{equation}
where $r_i\ (i = 1, 2, \dots, n)$ is the 2-norm of the $i$-th row of the matrix $A - \sigma I$:
\begin{equation}
  r_i \defeq \sqrt{|a_{ii} - \sigma|^2 + \sum_{k \neq i} |a_{ik}|^2}.
\end{equation}
We demonstrate the effectiveness of the proposed preconditioner in Sec.~\ref{sec:results}.

\section{Twist operator}
\label{sec:twist_operator}

The twist operator
\begin{equation}
    \twist \coloneqq \exp \left[ i \frac{2 \pi}{L} \sum_{j = 1}^L j \sz_{j} \right]
\end{equation}
is a unitary operator that generates a spin-wave-like excitation by rotating spins around the $z$ axis with angles $\theta_j \coloneqq \frac{2\pi}{L} j$ gradually increasing over sites $j = 1, 2, \dots, L$.
The twist operator appeared in the proof of the Lieb-Schultz-Mattis theorem~\cite{Lieb1961} to create a trial state orthogonal to the ground state and has excitation energy of $\mathcal{O}(L^{-1})$.
Nakamura and Todo utilized the twist operator as an order parameter to detect quantum phase transitions~\cite{Nakamura2002}.

We propose the following quantity for detecting the many-body localization transition between thermal and localized eigenstates:
\begin{equation}
    z \defeq \expval{\twist}{\psi},
\end{equation}
where $\ket{\psi}$ is an eigenstate of the Hamiltonian.
This quantity, \emph{twist overlap} $z$, measures how much the twisted state $\twist \ket{\psi}$ overlaps with the original eigenstate $\ket{\psi}$.
For thermal eigenstates, we expect that the twist operator generates a spin-wave-like excitation and creates a new state which is orthogonal to the original one, as in the case without randomness~\cite{Lieb1961}.
On the other hand, we expect that localized eigenstates are not affected by the long-wave-length perturbation by the twist operator but are responsive to modifications of the local degrees of freedom.
In the rest of this section, we discuss the excitation energy induced by the twist operator and the finite-size behavior of the twist overlap for thermal and localized eigenstates.

First, let us confirm that the energy difference between the twisted state $\twist \ket{\psi}$ and the original eigenstate $\ket{\psi}$ vanishes in the thermodynamic limit for the Hamiltonian~\eqref{eq:heisenberg_hamiltonian}.
We can calculate the energy difference as in the case without randomness~\cite{Lieb1961,Affleck1986} since the random magnetic-field terms in the Hamiltonian~\eqref{eq:heisenberg_hamiltonian} are invariant under the transformation by the twist operator.
By using the following identities,
\begin{align}
  & \conjtp{\twist} \sz_{j} \twist = \sz_{j}, \\
  & \conjtp{\twist} \splus_{j} \twist = e^{-i \theta_j} \splus_{j}, \\
  & \conjtp{\twist} \sminus_{j} \twist = e^{i \theta_j} \sminus_{j},
\end{align}
where $\theta_j \defeq \frac{2\pi}{L} j$, we can write
\begin{multline}
  \conjtp{\twist} \ham \twist - \ham \\
  = \frac{J}{2} \left( \cos \left( \frac{2 \pi}{L} \right) - 1 \right) \sum_{j = 1}^{L} (\splus_{j} \sminus_{j + 1} + \sminus_{j} \splus_{j + 1}) \\
  + i \frac{J}{2} \sin \left( \frac{2 \pi}{L} \right) \sum_{j = 1}^{L} (\splus_{j} \sminus_{j + 1} - \sminus_{j} \splus_{j + 1}).
\end{multline}
Let $\ket{\psi}$ be an eigenstate of the Hamiltonian \eqref{eq:heisenberg_hamiltonian}.
The Hamiltonian is real symmetric in the $\sz$ basis.
Hence, the expansion coefficients of the eigenstate $\ket{\psi}$ can be chosen to be real and so are those of $\splus_{j} \sminus_{j + 1} \ket{\psi}$ and $\sminus_{j} \splus_{j + 1} \ket{\psi}$.
As a result, we obtain the energy difference as follows:
\begin{multline}
  \expval{(\twist^{\dagger} \ham \twist - \ham)}{\psi} = \frac{J}{2} \left( \cos \left( \frac{2 \pi}{L} \right) - 1 \right) \\
  \times \sum_{j = 1}^{L} \expval{(\splus_{j} \sminus_{j + 1} + \sminus_{j} \splus_{j + 1})}{\psi}.
\end{multline}
An upper bound can be given as follows assuming the eigenstate $\ket{\psi}$ is normalized:
\begin{align}
  |\expval{(\twist^{\dagger} \ham \twist - \ham)}{\psi}|
  \le \frac{|J|}{4} \left( \frac{2 \pi}{L} \right)^2 L
  = \frac{|J| \pi^2}{L},
\end{align}
which vanishes in the limit $L \to \infty$.

Next, we discuss the behavior of the twist overlap $z$ in the thermal phase.
Let us expand a thermal eigenstate $\ket{\psi}$ as
\begin{equation}
  \ket{\psi} = \sum_{\alpha = 1}^{n} c_{\alpha} \ket{\alpha} \quad (c_{\alpha} \in \mathbb{R}),
\end{equation}
where $\{ \ket{\alpha} \}$ is the set of $\sz$ basis vectors in the $\sz = 0$ sector of dimension $n = \binom{L}{L/2}$.
As an approximation of extended eigenstates, we assume that $c_{\alpha}$ ($\alpha = 1, 2, \dots, n$) are drawn independently from the normal distribution with zero mean and variance $1/n$.
The variance comes from the normalization of $\ket{\psi}$.
Then, the average of the squared norm of the expectation value of the twist operator becomes (see appendix~\ref{appendix:twist_thermal_gaussian})
\begin{equation}
    \overline{|\expval{\twist}{\psi}|^2}
    = \frac{2}{n}, \label{eq:twist_thermal_gaussian}
\end{equation}
which decreases exponentially as the system size $L$ grows.

Lastly, we discuss the behavior of the twist overlap $z$ in the localized phase by perturbation theory.
We assume $h \gg J$ and divide the Hamiltonian \eqref{eq:heisenberg_hamiltonian} into an unperturbed part $\ham_{0}$ and a perturbation part $\hat{V}$ as follows:
\begin{align}
  \ham     &= \ham_{0} + \hat{V}, \\
  \ham_{0} &= - \sum_{j = 1}^{L} h_j \sz_{j}, \\
  \hat{V}  &= J \sum_{j = 1}^{L} \svec_{j} \cdot \svec_{j + 1}.
\end{align}
The unperturbed part $\ham_{0}$ is diagonal in the $\sz$ basis $\{ \ket{\alpha} \}$:
\begin{equation}
  \ham_{0} \ket{\alpha} = E_{\alpha}^{(0)} \ket{\alpha}.
\end{equation}
Let us consider the following eigenstate $\ket{\psi}$ of the perturbed Hamiltonian $\ham$:
\begin{equation}
  \ket{\psi} = \ket{\alpha} + \sum_{\beta \neq \alpha} c_{\beta}^{(1)} \ket{\beta} + \mathcal{O}(J^2),
\end{equation}
which is expressed as a linear combination of the unperturbed eigenstates $\{ \ket{\alpha} \}$.
We fix the coefficient of $\ket{\alpha}$ to be unity.
$c_{\beta}^{(1)}$ is the $\mathcal{O}(J)$ coefficient of $\ket{\beta}$ in $\ket{\psi}$:
\begin{equation}
  c_{\beta}^{(1)} = \frac{\mel{\beta}{\hat{V}}{\alpha}}{E_{\alpha}^{(0)} - E_{\beta}^{(0)}}.
\end{equation}
Then, we can evaluate the expectation value of the twist operator as follows (see appendix~\ref{appendix:twist_perturbation_bound}):
\begin{multline} \label{eq:twist_perturbation_bound}
  \left| \frac{\expval{\twist}{\psi}}{1 + \braket{\psi}^{(2)}} - \expval{\twist}{\alpha} \right| \\
  \le 2 \sin \left( \frac{\pi}{L} \right) \frac{\braket{\psi}^{(2)}}{1 + \braket{\psi}^{(2)}} + \mathcal{O}(J^3),
\end{multline}
where $\braket{\psi}^{(2)}$ is the $\mathcal{O}(J^2)$ term in the normalization constant $\braket{\psi}$:
\begin{equation}
  \braket{\psi}^{(2)} = \sum_{\beta \neq \alpha} \left| c_{\beta}^{(1)} \right|^{2}.
\end{equation}
The $\mathcal{O}(J^2)$ term in the right hand side of Eq.~\eqref{eq:twist_perturbation_bound} vanishes as $\mathcal{O}(L^{-1})$ in the limit $L \to \infty$.

\section{Results}
\label{sec:results}
\subsection{Benchmark}
We implement the outer Lanczos loop on a central processing unit (CPU).
The most computationally intensive part within the Lanczos loop is the matrix inversion using MINRES.
The computation of MINRES consists of matrix-vector products, vector inner products, vector additions, and scalar-vector multiplications, which are memory-bound and can be performed embarrassingly parallel.
To accelerate the matrix inversion, we offload the execution of the MINRES algorithm to a graphical processing unit (GPU).
A GPU is a parallel computer equipped with a massive number of computation cores and high memory bandwidth.
We use the NVIDIA GPU (V100 and A100) and its programming model CUDA~\cite{CUDAProgrammingGuide}.
Figure~\ref{fig:siminres_impl} shows the concept of our implementation.
\begin{figure}[tb]
  \includegraphics[width=\hsize]{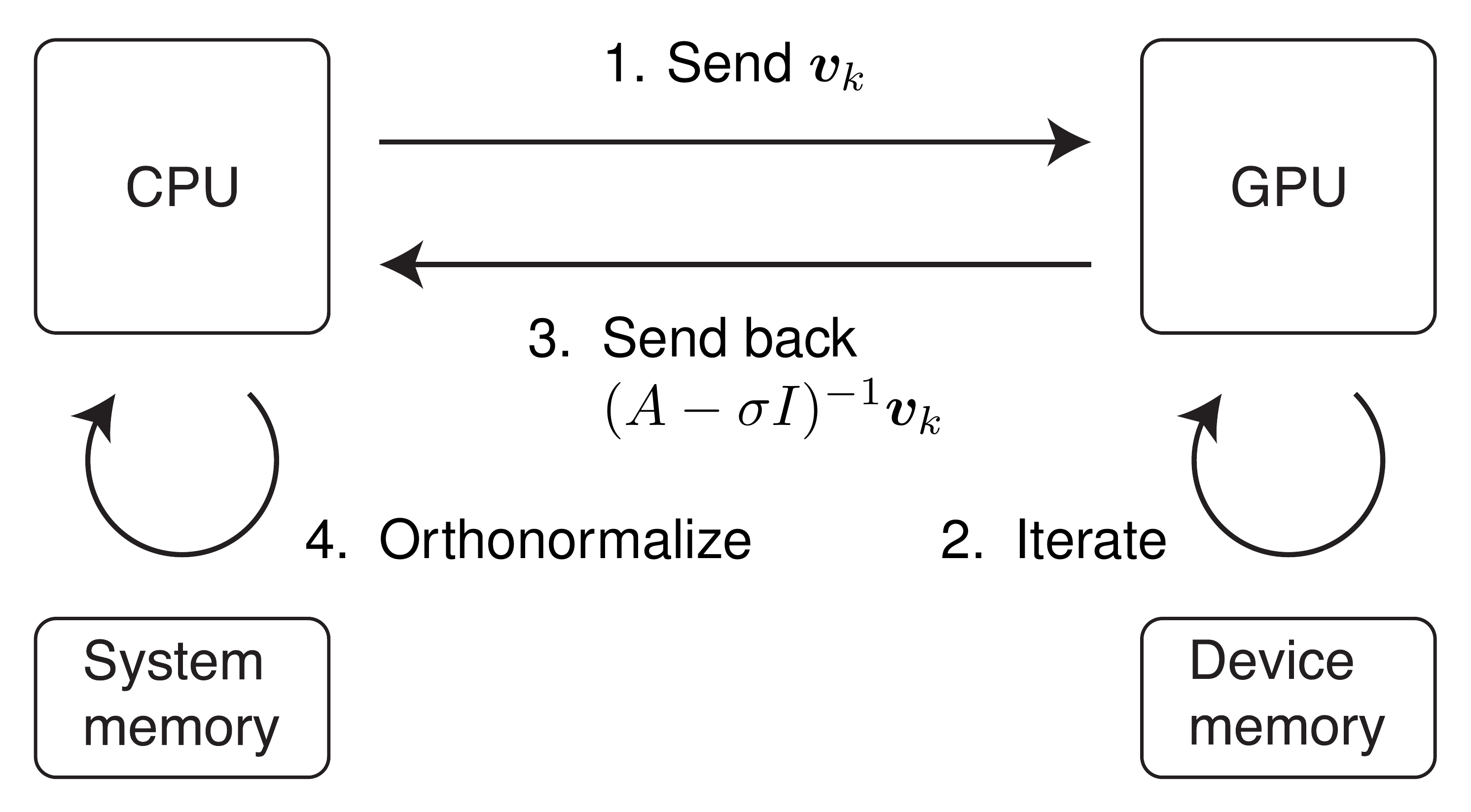}
  \caption{
    \label{fig:siminres_impl}
    Concept of our implementation of Alg.~\ref{alg:SIMINRES}.
    CPU and GPU have separate random access memories, which are called the system memory and device memory, respectively.
    We store the Lanczos basis vectors on the system memory.
    The vectors used by MINRES are stored on the device memory.
    First, we send a Lanczos basis vector $\bm{v}_k$ from the system memory to the device memory.
    Second, the GPU executes the iterations of MINRES to solve $(A - \sigma I) \bm{x} = \bm{v}_k$.
    Third, we send back the solution vector $\bm{x} = (A - \sigma I)^{-1} \bm{v}_k$ from the device memory to the system memory.
    Lastly, we obtain a new Lanczos basis vector $\bm{v}_{k + 1}$ by orthonormalizing $(A - \sigma I)^{-1} \bm{v}_k$ against previous Lanczos basis vectors on the CPU.
  }
\end{figure}

For the Heisenberg model's Hamiltonian \eqref{eq:heisenberg_hamiltonian}, we measure how much memory and time our program spends to obtain 50 eigenstates with eigenvalues closest to the target energy $\sigma = (E_{\text{min}} + E_{\text{max}}) / 2$, where $E_{\text{min}}$ and $E_{\text{max}}$ are the smallest and largest eigenvalues, respectively.
First, we measure the CPU and GPU memory for a single run of the program without any preconditioner.
We fix the disorder strength to be $h = 1$.
Figure~\ref{fig:siminres_memory} shows the result.
The memory usage is expected to be modeled as
\begin{align}
  f_{\rm CPU}(n) &\approx 8 \,\text{bytes} \times m \times n + C_{\rm CPU}, \label{eq:fcpu}\\
  f_{\rm GPU}(n) &\approx 8 \,\text{bytes} \times 9 \times n + C_{\rm GPU} \label{eq:fgpu}
\end{align}
for CPU and GPU, respectively. Here, $n$ is the dimension of the Hilbert space, \ie, the size of the matrix, and $m$ is the number of the iterations of the Lanczos algorithm (see Alg.~\ref{alg:SIMINRES}).
Factors 8 and 9 in Eqs.~\eqref{eq:fcpu} and \eqref{eq:fgpu} denote the size of double-precision floating number and the number of vectors required for MINRES, respectively.
We observe that the number of Lanczos iterations $m$ is almost constant regardless of the system size $L$ and the disorder strength $h$.
Thus, we fix it as a constant, $m=130$, for modeling the CPU memory usage by Eq.~\eqref{eq:fcpu}.
The constant terms $C_{\rm CPU}$ and $C_{\rm GPU}$ are determined from the memory usage measured for $L=12$ ($n=924$).
As shown in Fig.~\ref{fig:siminres_memory}, our linear memory-usage models describe the measured values very well, and we expect that the models are valid for larger $n$.
\begin{figure}[tb]
  \includegraphics[width=\hsize]{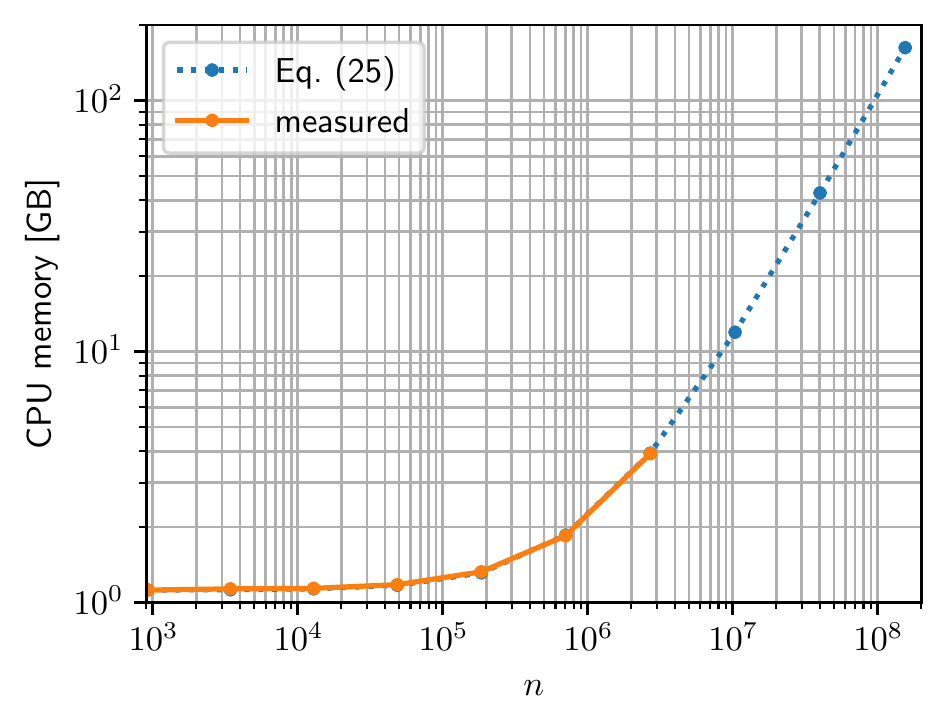}
  \includegraphics[width=\hsize]{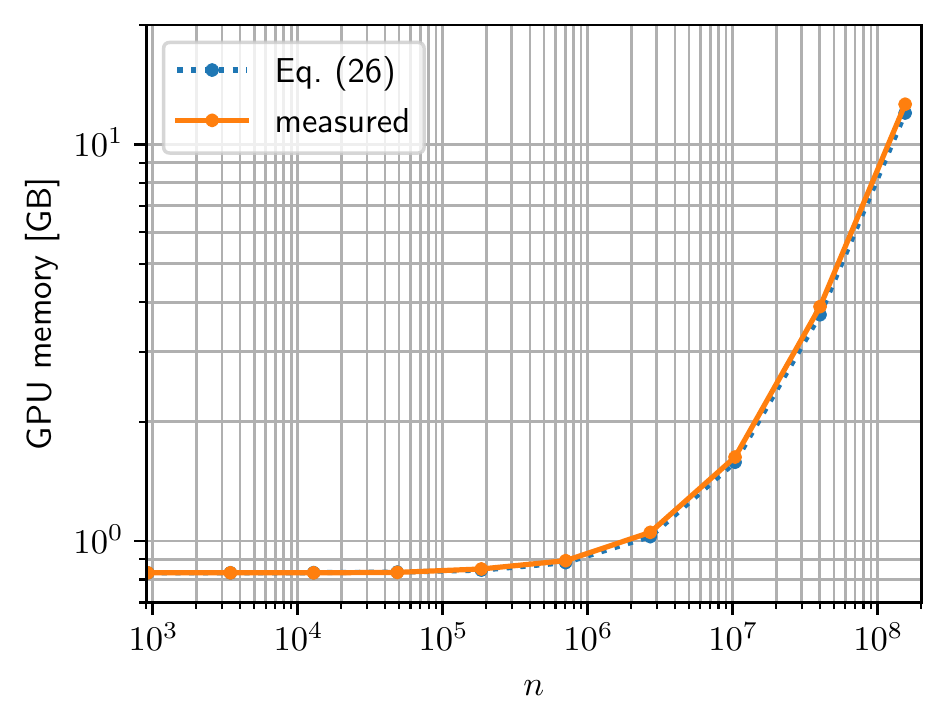}
  \caption{
    \label{fig:siminres_memory}
    Memory usage of CPU and GPU in Alg.~\ref{alg:SIMINRES} at $h=1$.
    The horizontal axis represents the dimension of the Hilbert space, \ie, the size of the matrix.
    The points in the figure correspond to $L=12, 14, 16, 18, 20, 22, 24, 26, 28, 30$, respectively.
    The measured memory usage is plotted by the orange solid lines.
    The blue dashed lines denote the linear memory-usage model described by Eqs.~\eqref{eq:fcpu} and \eqref{eq:fgpu}.
  }
\end{figure}

Our implementation of Alg.~\ref{alg:SIMINRES} uses several orders of magnitude less memory than the shift-invert Lanczos method with direct linear solvers.
In the latter, it was reported to be 244 GB for $L=22$~\cite{Pietracaprina2018}, while in the former is only 1.85 GB (CPU) plus 0.89 GB (GPU).
Even for system size $L=30$, we estimate that the memory usage of Alg.~\ref{alg:SIMINRES} is 162 GB (CPU) plus 12.6 GB (GPU), which fits well in modern workstations.
This significant reduction of memory usage is the essential advantage of the present algorithm.

Next, we measure the elapsed time (wall time) spent in a single program run.
We measure the execution time for the cases without preconditioner and with the row-norm preconditioner~\eqref{eq:rownorm_preconditioner}.
Figure~\ref{fig:siminres_time} shows the results for $h=3$ and 4.
The time complexity of the algorithm is expected to be $\mathcal{O}(n^2 L m)$.
As $m$ is almost constant regardless of $L$, the execution time is of order $n^2$ for large matrix sizes $n = 184\,756, 705\,432, 2\,704\,156$ ($L=20$, 22, 24, respectively).
The factor $L$ in the expected time complexity $\mathcal{O}(n^2 L m)$ can also be treated as a constant since $n$ increases exponentially with respect to $L$.
We also confirm that the row-norm preconditioner \eqref{eq:rownorm_preconditioner} significantly reduces the execution time even for $h = 3$ and 4, which is in the thermal phase as we see below.
\begin{figure}[tb]
  \includegraphics[width=\hsize]{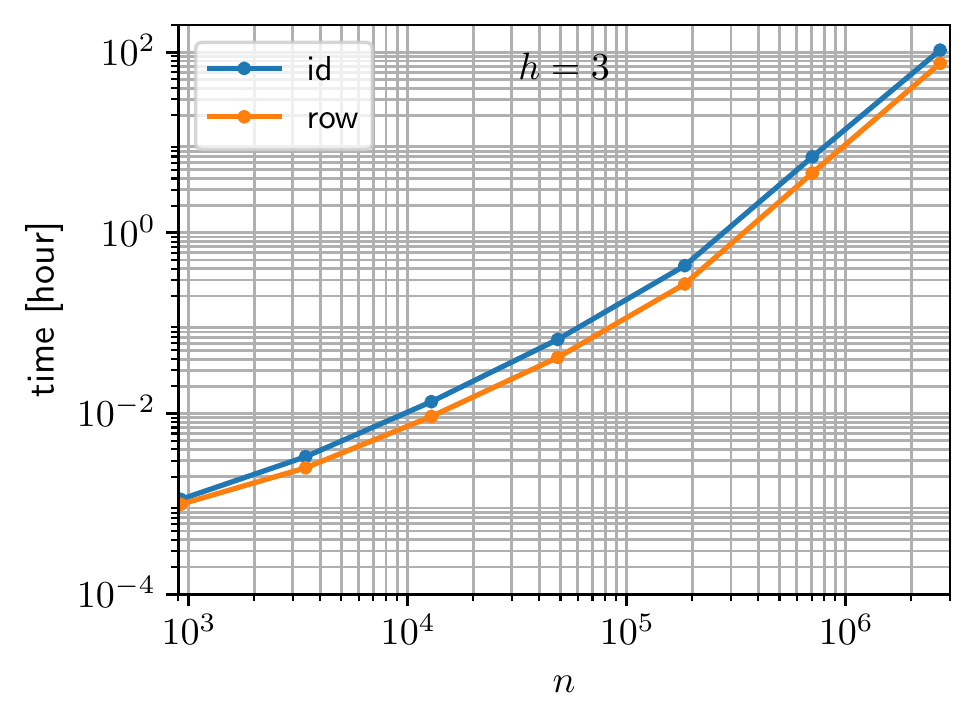}
  \includegraphics[width=\hsize]{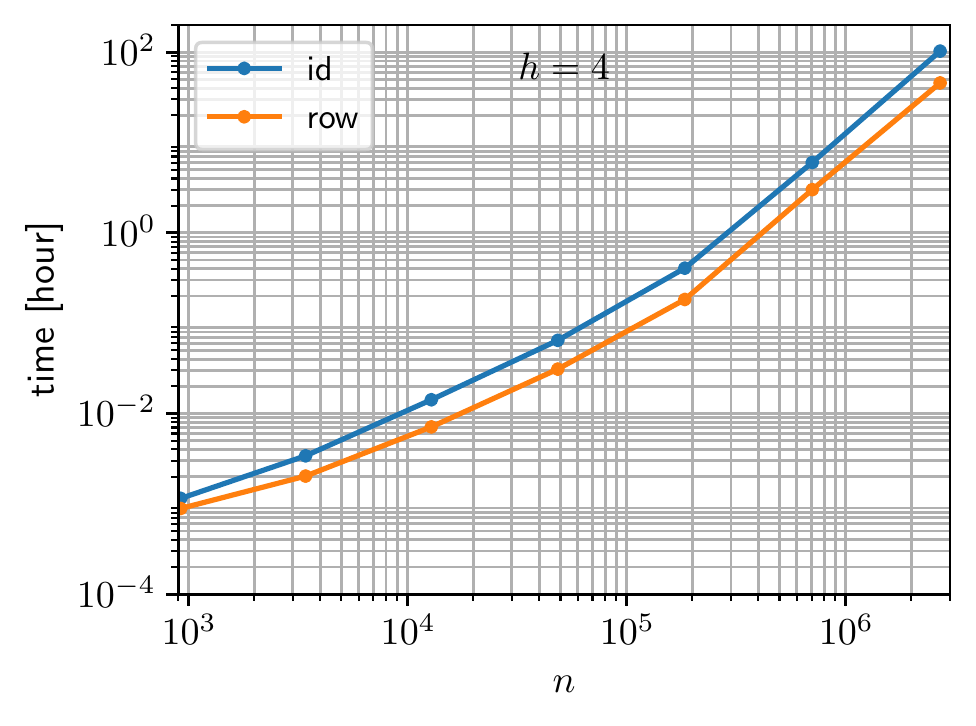}
  \caption{
    \label{fig:siminres_time}
    Elapsed time in a single run of Alg.~\ref{alg:SIMINRES} at $h=3$ and 4.
    The horizontal axis represents the dimension of the Hilbert space, \ie, the size of the matrix.
    The points in the figure correspond to $L=12, 14, 16, 18, 20, 22, 24$, respectively.
    The execution time is measured with two types of preconditioners: no preconditioner (id) and the row-norm preconditioner \eqref{eq:rownorm_preconditioner} (row).
  }
\end{figure}

\subsection{Entanglement and level spacing}
By using Alg.~\ref{alg:SIMINRES}, we calculate the bipartite entanglement entropy
\begin{equation}
  S = - \Tr \rho_A \log \rho_A
\end{equation}
and the level spacing ratio
\begin{equation}
  r_k = \frac{\min \{ \delta_k, \delta_{k - 1} \}}{\max \{ \delta_k, \delta_{k - 1} \}}
\end{equation}
for the random-field Heisenberg chain \eqref{eq:heisenberg_hamiltonian}.
Here, $\rho_A$ is the reduced density matrix for the subsystem of sites $j = 1, 2, \dots, L/2$.
The level spacing is defined as the difference between adjacent eigenvalues, $\delta_k = E_{k + 1} - E_{k}$.
For each pair of system size $L$ and disorder strength $h$, the number of disorder realizations is 10\,000 for $L \le 16$, 5\,000 for $L = 18$, 1\,000 for $L = 20$, and 100 for $L = 22$.
For each realization of disorder, 50 eigenstates with eigenvalues closest to the target energy $\sigma = (E_{\text{min}} + E_{\text{max}}) / 2$ are computed, where $E_{\text{min}}$ and $E_{\text{max}}$ are the smallest and largest eigenvalues, respectively.
The standard deviation is evaluated only over disorder realizations because eigenstates of the same disorder realization can be correlated.

Figure~\ref{fig:ee} shows the system-size dependence of the average bipartite entanglement entropy $\bar{S}$.
The entanglement $\bar{S}$ is confirmed to be proportional to the system volume $L$ for small $h$, which is expected behavior for thermal states.
We can see that the random-matrix-theory (RMT) value~\cite{Vidmar2017,Huang2019,Huang2021}
\begin{equation} \label{eq:ee_random}
   S_{\text{RMT}} (L) = \frac{\ln(2)}{2} L + \frac{1}{2} \ln \left( \frac{1}{2} \right) - \frac{1}{4}
\end{equation}
gives a decent estimate for the bipartite entanglement entropy for small enough $h$ even at the small system sizes $L \le 22$.
On the other hand, we can see that the entanglement $\bar{S}$ scales in a sub-volume manner for large enough $h$.
\begin{figure}[tb]
  \includegraphics[width=\hsize]{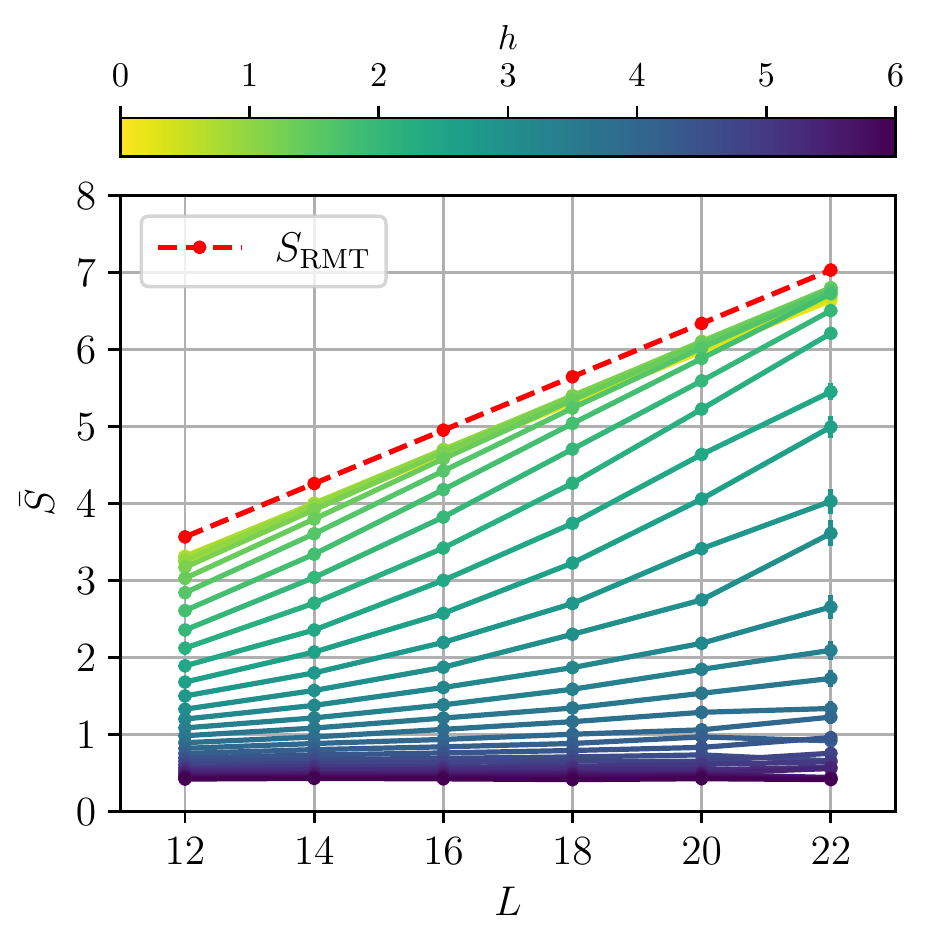}
  \caption{
    \label{fig:ee}
    System-size dependence of eigenstate entanglement entropy.
    The horizontal axis represents the system size.
    The vertical axis represents the average bipartite entanglement entropy.
    The RMT value \eqref{eq:ee_random} is plotted in a dashed line.
    The disorder strengths are $h=0.2, 0.4, 0.6, \dots, 5.8, 6.0$.
    The error bars are one standard deviation.
  }
\end{figure}

Figure~\ref{fig:lsr} shows the disorder-strength dependence of the average level spacing ratio $\bar{r}$.
We can confirm that the level spacing ratio $\bar{r}$ takes the Gaussian orthogonal ensemble (GOE) value $r_{\text{GOE}} = 0.5307(1)$~\cite{Atas2013} for small $h$.
We can also observe that the level spacing ratio $\bar{r}$ decreases down to the Poisson value $r_{\text{Poisson}} = 2 \ln(2) - 1 \approx 0.38629$ as the disorder strength $h$ increases.
\begin{figure}[tb]
  \includegraphics[width=\hsize]{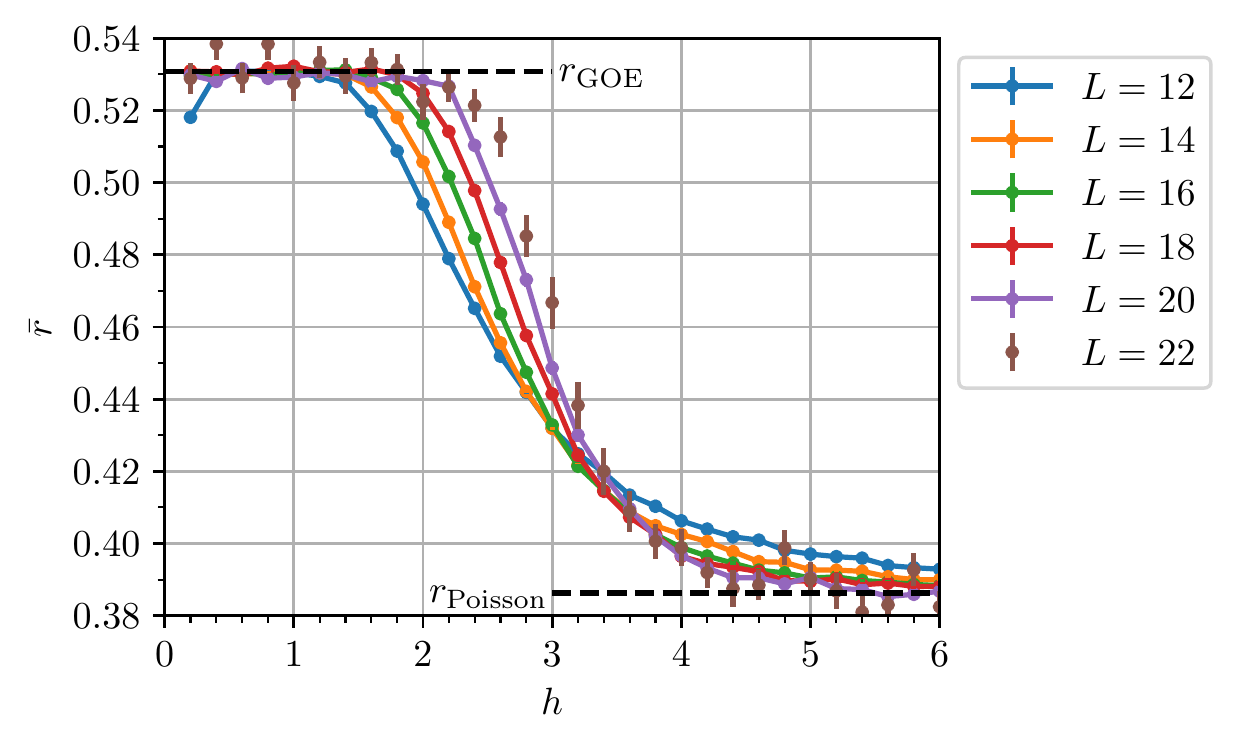}
  \caption{
    \label{fig:lsr}
    Disorder-strength dependence of level spacing ratio.
    The horizontal axis represents the disorder strength.
    The vertical axis represents the average level spacing ratio.
    The system size is shown in the legend.
    The error bars are one standard deviation.
  }
\end{figure}

\subsection{Twist operator}
We compute the eigenstate expectation value of the twist operator by using the same set of eigenvectors as Figs.~\ref{fig:ee} and \ref{fig:lsr}.
First, the squared norm of the expectation value of the twist operator is calculated for each eigenstate.
Second, the computed values are averaged over eigenstates of the same disorder realization.
Lastly, from these intra-spectrum-averaged values, we calculate the average and standard deviation.

Figure~\ref{fig:twist} shows the mean square $\zmeansq$ of the magnitudes of the twist overlap.
As the system size $L$ increases, we can observe that $\zmeansq$ decreases in the weak disorder $h$, but on the other hand, it increases in the strong disorder $h$.
The intersections of the curves of $\zmeansq$ for different system sizes $L$ drift slowly towards larger disorder strength $h$.
Compared with the level spacing ratio (Fig.~\ref{fig:lsr}), the disorder-strength dependence of the eigenstate expectation value of the twist operator is smoother and we can see the system-size dependence more clearly.
\begin{figure}[tb]
  \includegraphics[width=\hsize]{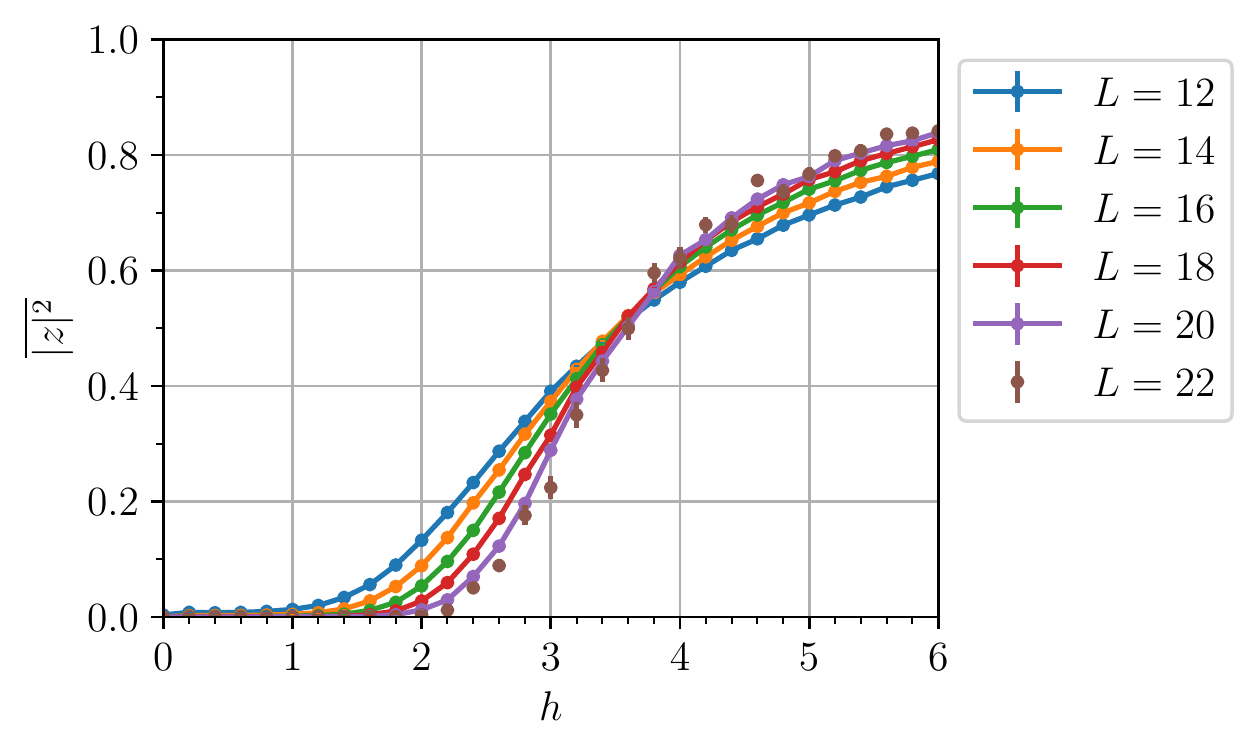}
  \caption{
    \label{fig:twist}
    Disorder-strength dependence of the eigenstate expectation value of the twist operator.
    The horizontal axis represents the disorder strength.
    The error bars are one standard deviation.
  }
\end{figure}

Let us examine the system-size scaling of $\zmeansq$ in the weak and strong disorder regimes.
Figure~\ref{fig:twist_thermal} shows $\zmeansq$ versus the dimension of the Hilbert space $n$ under the weak disorder.
It can be seen that $\zmeansq$ decreases with the power of $n$, \ie, exponentially with the system size $L$, for small $h$.
For $h=2$, the exponent of decay is evaluated as $-0.9$, which is close to the anticipated $n^{-1}$ scaling.
Figure~\ref{fig:twist_localized} shows the system-size dependence of $1 - \zmeansq$ in the strong disorder regime.
It can be seen that $1 - \zmeansq$ decreases with the power of $L$ for large $h$.
For $h=10$, the exponent of the power law is evaluated as $-0.9$, which is consistent with the predicted $L^{-1}$ scaling.
\begin{figure}[tb]
  \includegraphics[width=\hsize]{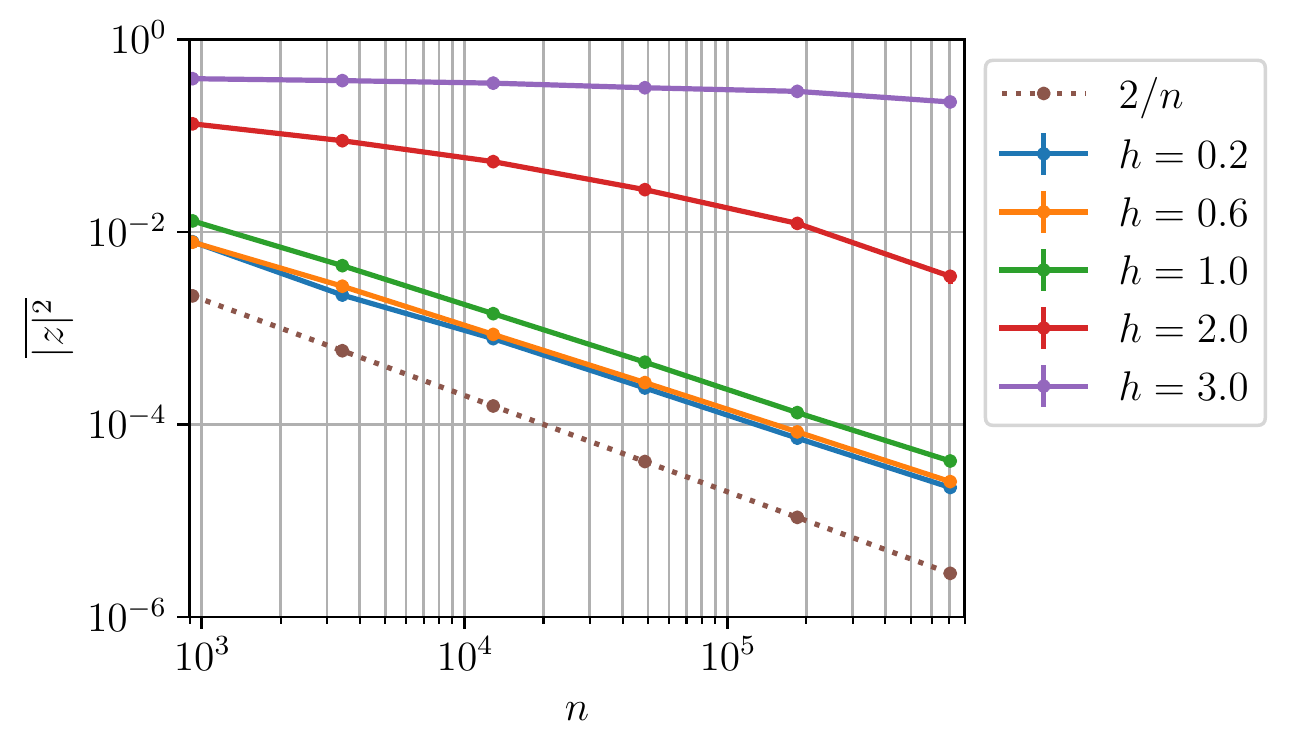}
  \caption{
    \label{fig:twist_thermal}
    System-size dependence of the eigenstate expectation values of the twist operator in the weak disorder regime.
    The horizontal axis represents the dimension $n$ of the Hilbert space.
    The system sizes corresponding to the points in the figure are $L=12, 14, 16, 18, 20, 22$.
    The dotted line indicates $2 / n$ for comparison.
  }
\end{figure}
\begin{figure}[tb]
  \includegraphics[width=\hsize]{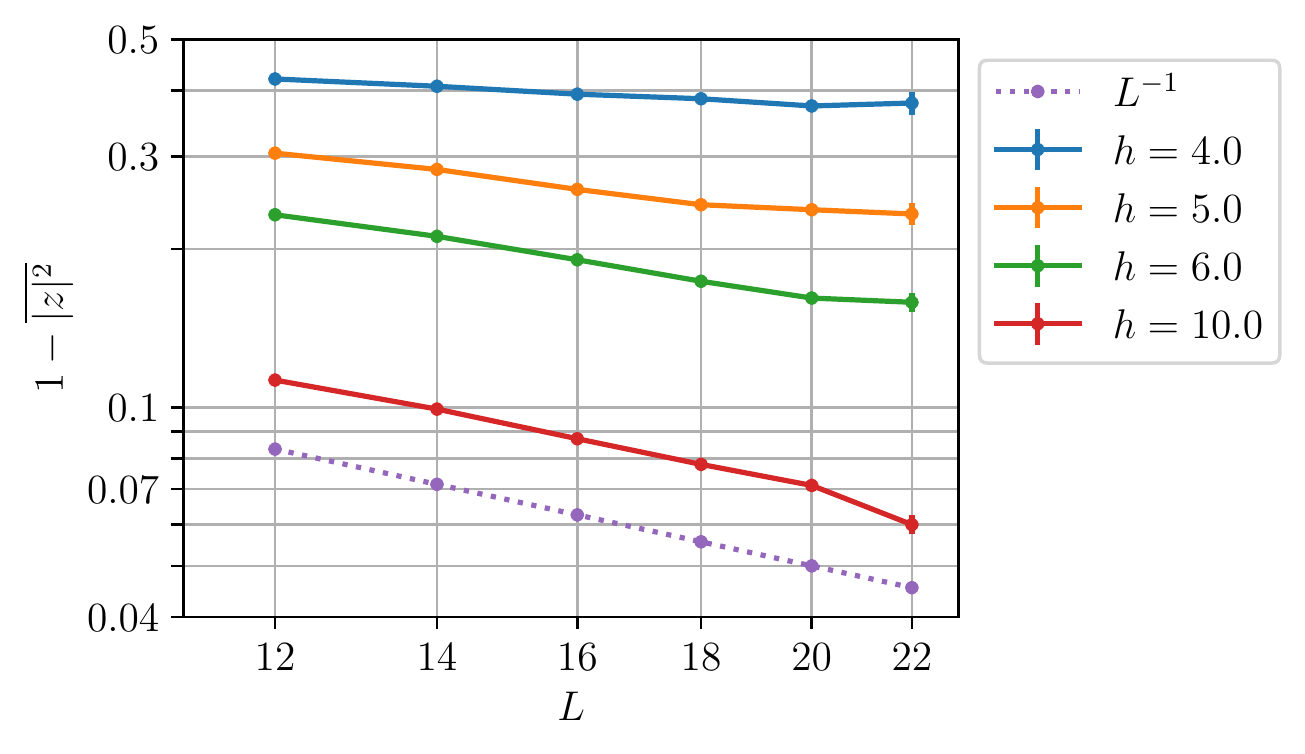}
  \caption{
    \label{fig:twist_localized}
    System-size dependence of the expectation value of the twist operator in the strong disorder regime.
    The horizontal axis represents the system size $L$.
    The dotted line indicates the $L^{-1}$ scaling for comparison.
  }
\end{figure}

Next, we estimate the many-body localization transition point of the model \eqref{eq:heisenberg_hamiltonian}.
In the same approach with a previous research~\cite{Sierant2020}, we interpolate $\zmeansq$ with a polynomial of degree three as a function of the disorder strength $h$ for each system size $L$.
Let $f_L(h)$ be the interpolation of $\zmeansq$ for system size $L$.
We calculate $h_c^1(L)$ as the intersection of interpolations $f_{L - 1}(h)$ and $f_{L + 1}(h)$.
In the same way, we define $h_c^2(L)$ as the intersection of interpolations $f_{L - 2}(h)$ and $f_{L + 2}(h)$.
As shown in Fig.~\ref{fig:twist_hc}, we fit $h_c^1(L)$ and $h_c^2(L)$ as linear functions of $L^{-1}$ and extrapolate them to the thermodynamic limit $L \to \infty$.
We thus conclude that the transition point is $h_c^{\text{twist}} \approx 5$.
This value is close to the value $h_c \approx 5.4$~\cite{Sierant2020} estimated by using the bipartite entanglement entropy and the level spacing ratio.
On the other hand, the present estimate is significantly greater than $h_c \approx 3.7$~\cite{Luitz2015} concluded by the finite-size data collapse.
According to Fig.~\ref{fig:twist_hc}, we can say that the disorder strength $h \approx 3.7$ is likely to be in the thermal region for larger systems.
\begin{figure}[tb]
  \includegraphics[width=\hsize]{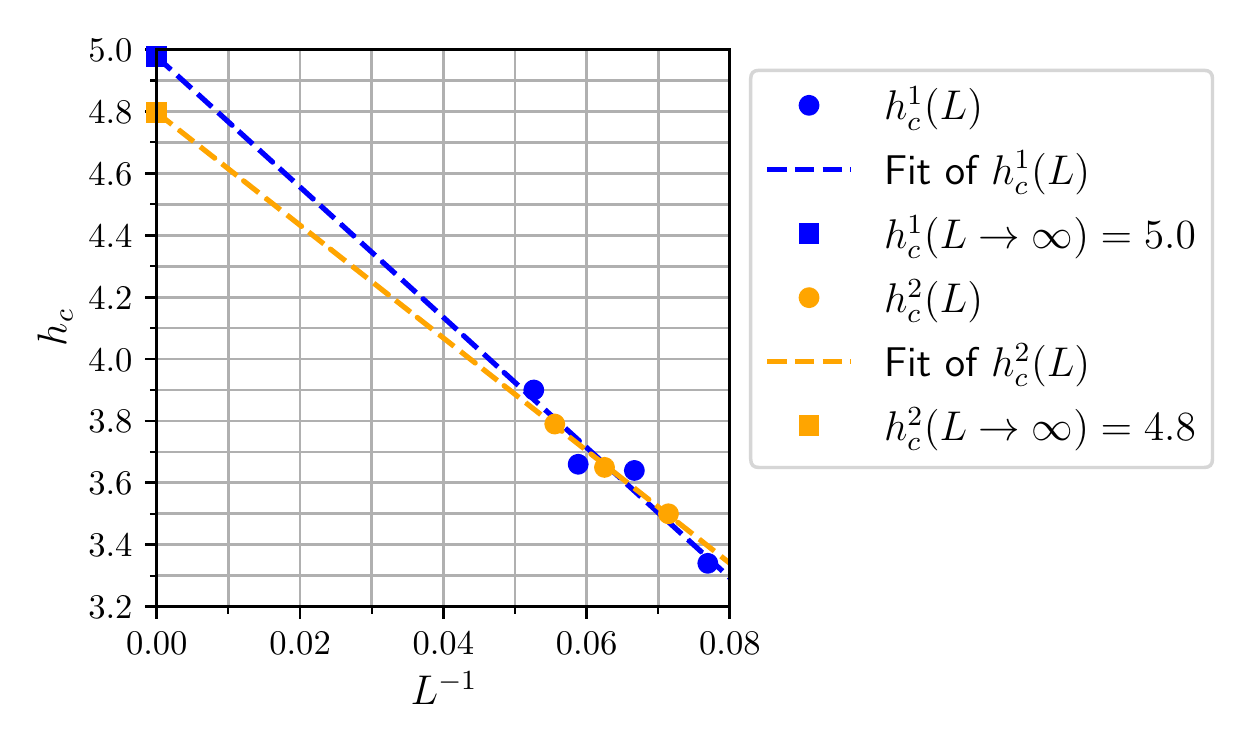}
  \caption{
    \label{fig:twist_hc}
    Estimation of many-body localization transition point of the random-field Heisenberg chain.
    The horizontal axis represents the inverse of the system size $L$.
    The system sizes corresponding to the points in the figure are $L=13, 14, 15, 16, 17, 18, 19$.
    We extrapolate $h_c^1(L)$ and $h_c^2(L)$ to the infinite system size $L$ as linear functions of $L^{-1}$.
  }
\end{figure}

\section{Conclusion}
\label{sec:conclusion}
In this paper, we proposed the nested iterative shift-invert Lanczos method with the minimum residual method (MINRES) for studying many-body localization.
As a result of substituting the iterative linear solver MINRES for direct linear solvers, the memory usage is reduced by several orders of magnitude, begin proportional to the dimension of a Hilbert space.
In addition, we introduced a simple preconditioner for the Hamiltonian of the random-field Heisenberg chain.
The preconditioner helps us to handle a large number of disorder realizations by decreasing the execution time of the algorithm.
We performed a large-scale exact diagonalization of the random-field Heisenberg chain with system size up to 22 and $10^2 \text{--} 10^4$ disorder realizations.
We consider that the algorithm can reach even larger systems with more problem-specific elaborate preconditioners and thus contribute to understanding the fate of the many-body localized phase in the thermodynamic limit.

Furthermore, we proposed the twist operator as a probe of many-body localization transition with a clear interpretation in the real space.
We discussed that the magnitude of the expectation value of the twist operator decreases for thermal eigenstates but increases for localized eigenstates as the system size grows.
We evaluated the twist overlap numerically for the random-field Heisenberg chain using the nested iterative shift-invert diagonalization method.
Also, we estimated the many-body localization transition point of the random-field Heisenberg chain from the system-size dependence of the mean twist overlap and obtained $h_c^{\text{twist}} \approx 5$.
It is left as future work to examine further the twist operator's capability to detect many-body localization transitions in other strongly-correlated quantum systems.

\begin{acknowledgments}
  The computation was performed on the workstations at the Institute for Physics of Intelligence ($i\pi$), the University of Tokyo, and the supercomputing system MASAMUNE-IMR at the Center for Computational Materials Science, the Tohoku University.
\end{acknowledgments}

\appendix

\section{\label{appendix:minres}Minimum residual method}
The minimum residual method (MINRES)~\cite{Paige1975} is an iterative algorithm to solve a system of linear equations $A \bm{x} = \bm{b}$ where $A \in \mathbb{C}^{n \times n}$ is an $n \times n$ Hermitian indefinite matrix and $\bm{b} \in \mathbb{C}^{n}$ a vector.
Let $\bm{x}_0 \in \mathbb{C}^n$ be the initial guess of the solution, $\bm{r}_0 \defeq \bm{b} - A \bm{x}_0$ the initial residual vector, and
\begin{equation}
  \mathcal{K}_{k}(A, \bm{r}_0) \defeq \text{span} (\bm{r}_0, A \bm{r}_0, A^2 \bm{r}_0, \dots, A^{k - 1} \bm{r}_0) \subset \mathbb{C}^{n}
\end{equation}
the $k$-dimensional Krylov subspace.
The solution $\bm{x}_k$ that MINRES yields at the $k$-th step of the algorithm is the least-square solution within the $k$-dimensional Krylov subspace translated by $\bm{x}_0$:
\begin{equation} \label{eq:minres_solution}
  \bm{x}_k \defeq \argmin_{\bm{x} \in \bm{x}_0 + \mathcal{K}_{k} (A, \bm{r}_0)} \norm{\bm{b} - A \bm{x}}.
\end{equation}
This least-square problem can be solved by reducing the coefficient matrix $A$ into a tridiagonal matrix with the Lanczos three-term recurrence relation and eliminating the subdiagonals by a simple Givens rotation~\cite{TemplatesLinear}.
Algorithm~\ref{alg:PMINRES} shows the pseudocode of MINRES that solves a preconditioned system
\begin{equation}
  (C^{-1} A \conjtp{(C^{-1})})(\conjtp{C} \bm{x}) = C^{-1} \bm{v}.
\end{equation}
The preconditioning matrix is defined as $M \coloneqq C \conjtp{C}$.
\begin{figure}[tb]
  \begin{algorithm}[H]
    \caption{Preconditioned minimum residual method} \label{alg:PMINRES}
    \begin{algorithmic}[1]
      \Function{MINRES}{$A,\ \bm{x}_0,\ \bm{b},\ M,\ k_\text{max},\ \text{tolerance}$}
      \State $\bm{z}_0 \gets \bm{0},\ \bm{z}_1 \gets \bm{b} - A \bm{x}_0$
      \State $\bm{q}_1 \gets M^{-1} \bm{z}_1$
      \State $\beta_0 \gets 1,\ \beta_1 \gets \sqrt{\conjtp{\bm{q}_1} \bm{z}_1}$
      \State $\phi_0 \gets \beta_1,\ \hat{\delta}_1 \gets 0$
      \State $c_0 \gets 1,\ s_0 \gets 0$
      \State $\bm{d}_{-1} \gets \bm{0},\ \bm{d}_{0} \gets \bm{0}$
      \State $k \gets 0$
      \While{$|\phi_k| \ge \text{tolerance}$ and $k < k_\text{max}$}
        \State $k \gets k + 1$
        \State $\bm{p}_k \gets A \bm{q}_k$
        \State $\alpha_k \gets \conjtp{\bm{q}_k} \bm{p}_k / \beta_k^2$
        \State $\bm{z}_{k + 1} \gets \frac{1}{\beta_k} \bm{p}_k - \frac{\alpha_k}{\beta_k} \bm{z}_k - \frac{\beta_k}{\beta_{k - 1}} \bm{z}_{k - 1}$
        \State $\bm{q}_{k + 1} \gets M^{-1} \bm{z}_{k + 1}$
        \State $\beta_{k + 1} \gets \sqrt{\conjtp{\bm{q}_{k + 1}} \bm{z}_{k + 1}}$
        \State $\delta_k \gets c_{k - 1} \hat{\delta}_k + s_{k - 1} \alpha_k$
        \State $\hat{\gamma}_k \gets -s_{k - 1} \hat{\delta}_k + c_{k - 1} \alpha_k$
        \State $\epsilon_{k + 1} \gets s_{k - 1} \beta_{k + 1}$
        \State $\hat{\delta}_{k + 1} \gets c_{k - 1} \beta_{k + 1}$
        \State $\gamma_{k} \gets \sqrt{\hat{\gamma}_{k}^2 + \beta_{k + 1}^2}$
        \State $c_k \gets \hat{\gamma}_{k} / \gamma_{k}$
        \State $s_k \gets \beta_{k + 1} / \gamma_{k}$
        \State $\tau_k \gets c_k \phi_{k - 1}$
        \State $\phi_k \gets -s_k \phi_{k - 1}$
        \State $\bm{d}_k \gets (\frac{1}{\beta_k} \bm{q}_k - \delta_k \bm{d}_{k - 1} - \epsilon_k \bm{d}_{k - 2}) / \gamma_k$
        \State $\bm{x}_k \gets \bm{x}_{k - 1} + \tau_k \bm{d}_k$
      \EndWhile
      \State \Return $\bm{x}_k$
      \EndFunction
    \end{algorithmic}
  \end{algorithm}
\end{figure}

\section{\label{appendix:twist_thermal_gaussian}Derivation of Eq.~\eqref{eq:twist_thermal_gaussian}}
The expectation value of the twist operator can be written as follows:
\begin{equation}
  \expval{\twist}{\psi} = \sum_{\alpha = 1}^{n} c_{\alpha}^2 u_{\alpha} \quad (u_{\alpha} \defeq \expval{\twist}{\alpha}),
\end{equation}
whose squared norm becomes
\begin{equation}
  |\expval{\twist}{\psi}|^2
  = \sum_{\alpha = 1}^{n} c_{\alpha}^4 + \sum_{\alpha \neq \beta} u_{\alpha}^{*} u_{\beta} c_{\alpha}^2 c_{\beta}^2,
\end{equation}
where we used $|u_{\alpha}| = 1$.
The average of the squared norm can be calculated as follows:
\begin{align}
  \begin{split}
    \overline{|\expval{\twist}{\psi}|^2}
    &= \sum_{\alpha = 1}^{n} \overline{c_{\alpha}^4} + \sum_{\alpha \neq \beta} u_{\alpha}^{*} u_{\beta} \overline{c_{\alpha}^2 c_{\beta}^2} \\
    &= \sum_{\alpha = 1}^{n} \overline{c_{\alpha}^4} + \sum_{\alpha \neq \beta} u_{\alpha}^{*} u_{\beta} \left( \overline{c_{\alpha}^2} \right) \left( \overline{c_{\beta}^2} \right) \\
    &= n \frac{3}{n^2} + \frac{1}{n^2} \sum_{\alpha \neq \beta} u_{\alpha}^{*} u_{\beta}
    = \frac{2}{n}.
  \end{split}
\end{align}
Here we used the following identity:
\begin{equation}
  n + \sum_{\alpha \neq \beta} u_{\alpha}^{*} u_{\beta} = 0,
\end{equation}
which can be obtained by taking the squared norm of the both sides of Eq.~\eqref{eq:twist_sum} in the following lemma.
\begin{lemma}
  \begin{equation}
    \sum_{\alpha = 1}^{n} u_{\alpha} = 0. \label{eq:twist_sum}
  \end{equation}
\end{lemma}
\begin{proof}
  We consider the following state:
  \begin{equation}
    \ket{\phi} = \frac{1}{\sqrt{n}} \sum_{\alpha = 1}^{n} \ket{\alpha},
  \end{equation}
  for which the expectation value of the twist operator is
  \begin{equation} \label{eq:twist_sum_phi}
    \expval{\twist}{\phi} = \frac{1}{n} \sum_{\alpha = 1}^{n} u_{\alpha}.
  \end{equation}
  Let $\hat{T}$ be the translation operator which satisfies
  \begin{equation}
    \conjtp{\hat{T}} \svec_{j} \hat{T} = \svec_{j - 1},
  \end{equation}
  where we identify $\svec_{0}$ with $\svec_{L}$.
  The state $\ket{\psi}$ is translation invariant, hence we have
  \begin{equation} \label{eq:twist_sum_translation_invariant}
    \expval{\conjtp{\hat{T}} \twist \hat{T}}{\phi} = \expval{\twist}{\phi}.
  \end{equation}
  On the other hand, if we act the translation operator on the twist operator, we obtain
  \begin{align}
    \begin{split}
      \conjtp{\hat{T}} \twist \hat{T}
      &= \exp \left[i \frac{2\pi}{L} \sum_{j = 1}^{L} j \sz_{j - 1} \right] \\
      &= \exp \left[i \frac{2\pi}{L} \sum_{j = 1}^{L} \{ (j - 1) + 1 \} \sz_{j - 1} \right] \\
      &= \exp \left[i \frac{2\pi}{L} \sum_{j = 1}^{L} j \sz_{j} \right] e^{-i \frac{2\pi}{L} L \sz_{L}} e^{i \frac{2\pi}{L} \sum_{j = 1}^{L} \sz_{j}} \\
      &= -\twist,
    \end{split}
  \end{align}
  where we used $e^{-2i\pi \sz_{L}} = -1$ and $\sum_{j = 1}^{L} \sz_{j} = 0$.
  Therefore, we also have
  \begin{equation} \label{eq:twist_sum_twist_translation}
    \expval{\conjtp{\hat{T}} \twist \hat{T}}{\phi} = -\expval{\twist}{\phi}.
  \end{equation}
  Eqs.~\eqref{eq:twist_sum_phi}, \eqref{eq:twist_sum_translation_invariant}, and \eqref{eq:twist_sum_twist_translation} give Eq.~\eqref{eq:twist_sum}.
\end{proof}

\section{\label{appendix:twist_perturbation_bound}Derivation of Eq.~\eqref{eq:twist_perturbation_bound}}
Since the twist operator $\twist$ is diagonal with respect to the unperturbed eigenstates, we can write the expectation value of the twist operator as follows:
\begin{multline} \label{eq:twist_2}
  \expval{\twist}{\psi} = \expval{\twist}{\alpha} \\
  + \sum_{\beta \neq \alpha} \left| c_{\beta}^{(1)} \right|^{2} \expval{\twist}{\beta} + \mathcal{O}(J^3).
\end{multline}
Let us evaluate the leading $\mathcal{O}(J^2)$ perturbation.
The off-diagonal element $\mel{\beta}{\hat{V}}{\alpha}$ is nonzero if and only if there is a site $j$ such that the $j$-th and $(j + 1)$-th spins of $\ket{\alpha}$ are opposite and $\ket{\beta}$ is the same with $\ket{\alpha}$ except the $j$-th and $(j + 1)$-th spins are flipped.
With this flip operation, the expectation value of the twist operator gains a phase factor $e^{\pm i \frac{2\pi}{L}}$.
For example, if the $j$-th and $(j + 1)$-th spins of $\ket{\alpha}$ is up and down, respectively, and $\ket{\beta}$ is $\ket{\alpha}$ with the two spins flipped, then we have
\begin{equation}
  \expval{\twist}{\beta} = \expval{\twist}{\alpha} e^{+i \frac{2\pi}{L}}.
\end{equation}
Let us define an integer $d_{\alpha \beta}\ (\beta \neq \alpha)$ as follows:
\begin{equation}
  \expval{\twist}{\beta} = \expval{\twist}{\alpha} e^{i \frac{2\pi}{L} d_{\alpha \beta}},
\end{equation}
then the integer $d_{\alpha \beta}$ satisfies $|d_{\alpha \beta}| = 1$ if the expansion coefficient $c_{\beta}^{(1)}$ is nonzero.
With this integer $d_{\alpha \beta}$, Eq.~\eqref{eq:twist_2} is rewritten as follows:
\begin{align}
  \begin{split}
  & \expval{\twist}{\psi} \\
  &= \expval{\twist}{\alpha} \left( 1 + \sum_{\beta \neq \alpha} \left| c_{\beta}^{(1)} \right|^{2} e^{i \frac{2\pi}{L} d_{\alpha \beta}} \right) + \mathcal{O}(J^3) \\
  &= \expval{\twist}{\alpha} \\
  &\times \left( 1 + \braket{\psi}^{(2)} + \sum_{\beta \neq \alpha} \left| c_{\beta}^{(1)} \right|^{2} (e^{i \frac{2\pi}{L} d_{\alpha \beta}} - 1) \right) + \mathcal{O}(J^3).
  \end{split}
\end{align}
We can write
\begin{align}
  \begin{split}
    & \left| \expval{\twist}{\psi} - \expval{\twist}{\alpha} \left( 1 + \braket{\psi}^{(2)} \right) \right| \\
    &= \left| \sum_{\beta \neq \alpha} \left| c_{\beta}^{(1)} \right|^{2} (e^{i \frac{2\pi}{L} d_{\alpha \beta}} - 1) \right| + \mathcal{O}(J^3) \\
    &\le \sum_{\beta \neq \alpha} \left| c_{\beta}^{(1)} \right|^{2} \left| e^{i \frac{2\pi}{L} d_{\alpha \beta}} - 1 \right| + \mathcal{O}(J^3) \\
    &= \sum_{\beta \neq \alpha} \left| c_{\beta}^{(1)} \right|^{2} \left| 2 \sin \left( \frac{\pi}{L} d_{\alpha \beta} \right) \right| + \mathcal{O}(J^3) \\
    &= 2 \sin \left( \frac{\pi}{L} \right) \braket{\psi}^{(2)} + \mathcal{O}(J^3),
  \end{split}
\end{align}
from which we obtain Eq.~\eqref{eq:twist_perturbation_bound}.

\bibliography{main}

%apsrev4-2.bst 2019-01-14 (MD) hand-edited version of apsrev4-1.bst
%Control: key (0)
%Control: author (72) initials jnrlst
%Control: editor formatted (1) identically to author
%Control: production of article title (-1) disabled
%Control: page (0) single
%Control: year (1) truncated
%Control: production of eprint (0) enabled
\begin{thebibliography}{55}%
\makeatletter
\providecommand \@ifxundefined [1]{%
 \@ifx{#1\undefined}
}%
\providecommand \@ifnum [1]{%
 \ifnum #1\expandafter \@firstoftwo
 \else \expandafter \@secondoftwo
 \fi
}%
\providecommand \@ifx [1]{%
 \ifx #1\expandafter \@firstoftwo
 \else \expandafter \@secondoftwo
 \fi
}%
\providecommand \natexlab [1]{#1}%
\providecommand \enquote  [1]{``#1''}%
\providecommand \bibnamefont  [1]{#1}%
\providecommand \bibfnamefont [1]{#1}%
\providecommand \citenamefont [1]{#1}%
\providecommand \href@noop [0]{\@secondoftwo}%
\providecommand \href [0]{\begingroup \@sanitize@url \@href}%
\providecommand \@href[1]{\@@startlink{#1}\@@href}%
\providecommand \@@href[1]{\endgroup#1\@@endlink}%
\providecommand \@sanitize@url [0]{\catcode `\\12\catcode `\$12\catcode
  `\&12\catcode `\#12\catcode `\^12\catcode `\_12\catcode `\%12\relax}%
\providecommand \@@startlink[1]{}%
\providecommand \@@endlink[0]{}%
\providecommand \url  [0]{\begingroup\@sanitize@url \@url }%
\providecommand \@url [1]{\endgroup\@href {#1}{\urlprefix }}%
\providecommand \urlprefix  [0]{URL }%
\providecommand \Eprint [0]{\href }%
\providecommand \doibase [0]{https://doi.org/}%
\providecommand \selectlanguage [0]{\@gobble}%
\providecommand \bibinfo  [0]{\@secondoftwo}%
\providecommand \bibfield  [0]{\@secondoftwo}%
\providecommand \translation [1]{[#1]}%
\providecommand \BibitemOpen [0]{}%
\providecommand \bibitemStop [0]{}%
\providecommand \bibitemNoStop [0]{.\EOS\space}%
\providecommand \EOS [0]{\spacefactor3000\relax}%
\providecommand \BibitemShut  [1]{\csname bibitem#1\endcsname}%
\let\auto@bib@innerbib\@empty
%</preamble>
\bibitem [{\citenamefont {D'Alessio}\ \emph {et~al.}(2016)\citenamefont
  {D'Alessio}, \citenamefont {Kafri}, \citenamefont {Polkovnikov},\ and\
  \citenamefont {Rigol}}]{Alessio2016}%
  \BibitemOpen
  \bibfield  {author} {\bibinfo {author} {\bibfnamefont {L.}~\bibnamefont
  {D'Alessio}}, \bibinfo {author} {\bibfnamefont {Y.}~\bibnamefont {Kafri}},
  \bibinfo {author} {\bibfnamefont {A.}~\bibnamefont {Polkovnikov}},\ and\
  \bibinfo {author} {\bibfnamefont {M.}~\bibnamefont {Rigol}},\ }\href
  {https://doi.org/10.1080/00018732.2016.1198134} {\bibfield  {journal}
  {\bibinfo  {journal} {Adv. Phys.}\ }\textbf {\bibinfo {volume} {65}},\
  \bibinfo {pages} {239} (\bibinfo {year} {2016})}\BibitemShut {NoStop}%
\bibitem [{\citenamefont {Deutsch}(1991)}]{Deutsch1991}%
  \BibitemOpen
  \bibfield  {author} {\bibinfo {author} {\bibfnamefont {J.~M.}\ \bibnamefont
  {Deutsch}},\ }\href {https://doi.org/10.1103/PhysRevA.43.2046} {\bibfield
  {journal} {\bibinfo  {journal} {Phys. Rev. A}\ }\textbf {\bibinfo {volume}
  {43}},\ \bibinfo {pages} {2046} (\bibinfo {year} {1991})}\BibitemShut
  {NoStop}%
\bibitem [{\citenamefont {Srednicki}(1994)}]{Srednicki1994}%
  \BibitemOpen
  \bibfield  {author} {\bibinfo {author} {\bibfnamefont {M.}~\bibnamefont
  {Srednicki}},\ }\href {https://doi.org/10.1103/PhysRevE.50.888} {\bibfield
  {journal} {\bibinfo  {journal} {Phys. Rev. E}\ }\textbf {\bibinfo {volume}
  {50}},\ \bibinfo {pages} {888} (\bibinfo {year} {1994})}\BibitemShut
  {NoStop}%
\bibitem [{\citenamefont {Srednicki}(1999)}]{Srednicki1999}%
  \BibitemOpen
  \bibfield  {author} {\bibinfo {author} {\bibfnamefont {M.}~\bibnamefont
  {Srednicki}},\ }\href {https://doi.org/10.1088/0305-4470/32/7/007} {\bibfield
   {journal} {\bibinfo  {journal} {J. Phys. A: Math. Gen.}\ }\textbf {\bibinfo
  {volume} {32}},\ \bibinfo {pages} {1163} (\bibinfo {year}
  {1999})}\BibitemShut {NoStop}%
\bibitem [{\citenamefont {Rigol}\ \emph {et~al.}(2008)\citenamefont {Rigol},
  \citenamefont {Dunjko},\ and\ \citenamefont {Olshanii}}]{Rigol2008}%
  \BibitemOpen
  \bibfield  {author} {\bibinfo {author} {\bibfnamefont {M.}~\bibnamefont
  {Rigol}}, \bibinfo {author} {\bibfnamefont {V.}~\bibnamefont {Dunjko}},\ and\
  \bibinfo {author} {\bibfnamefont {M.}~\bibnamefont {Olshanii}},\ }\href
  {https://doi.org/10.1038/nature06838} {\bibfield  {journal} {\bibinfo
  {journal} {Nature}\ }\textbf {\bibinfo {volume} {452}},\ \bibinfo {pages}
  {854} (\bibinfo {year} {2008})}\BibitemShut {NoStop}%
\bibitem [{\citenamefont {Rigol}(2009)}]{Rigol2009}%
  \BibitemOpen
  \bibfield  {author} {\bibinfo {author} {\bibfnamefont {M.}~\bibnamefont
  {Rigol}},\ }\href {https://doi.org/10.1103/PhysRevLett.103.100403} {\bibfield
   {journal} {\bibinfo  {journal} {Phys. Rev. Lett.}\ }\textbf {\bibinfo
  {volume} {103}},\ \bibinfo {pages} {100403} (\bibinfo {year}
  {2009})}\BibitemShut {NoStop}%
\bibitem [{\citenamefont {Santos}\ and\ \citenamefont
  {Rigol}(2010)}]{Santos2010}%
  \BibitemOpen
  \bibfield  {author} {\bibinfo {author} {\bibfnamefont {L.~F.}\ \bibnamefont
  {Santos}}\ and\ \bibinfo {author} {\bibfnamefont {M.}~\bibnamefont {Rigol}},\
  }\href {https://doi.org/10.1103/PhysRevE.82.031130} {\bibfield  {journal}
  {\bibinfo  {journal} {Phys. Rev. E}\ }\textbf {\bibinfo {volume} {82}},\
  \bibinfo {pages} {031130} (\bibinfo {year} {2010})}\BibitemShut {NoStop}%
\bibitem [{\citenamefont {Biroli}\ \emph {et~al.}(2010)\citenamefont {Biroli},
  \citenamefont {Kollath},\ and\ \citenamefont {L\"auchli}}]{Biroli2010}%
  \BibitemOpen
  \bibfield  {author} {\bibinfo {author} {\bibfnamefont {G.}~\bibnamefont
  {Biroli}}, \bibinfo {author} {\bibfnamefont {C.}~\bibnamefont {Kollath}},\
  and\ \bibinfo {author} {\bibfnamefont {A.~M.}\ \bibnamefont {L\"auchli}},\
  }\href {https://doi.org/10.1103/PhysRevLett.105.250401} {\bibfield  {journal}
  {\bibinfo  {journal} {Phys. Rev. Lett.}\ }\textbf {\bibinfo {volume} {105}},\
  \bibinfo {pages} {250401} (\bibinfo {year} {2010})}\BibitemShut {NoStop}%
\bibitem [{\citenamefont {Steinigeweg}\ \emph {et~al.}(2013)\citenamefont
  {Steinigeweg}, \citenamefont {Herbrych},\ and\ \citenamefont
  {Prelov\ifmmode~\check{s}\else \v{s}\fi{}ek}}]{Steinigeweg2013}%
  \BibitemOpen
  \bibfield  {author} {\bibinfo {author} {\bibfnamefont {R.}~\bibnamefont
  {Steinigeweg}}, \bibinfo {author} {\bibfnamefont {J.}~\bibnamefont
  {Herbrych}},\ and\ \bibinfo {author} {\bibfnamefont {P.}~\bibnamefont
  {Prelov\ifmmode~\check{s}\else \v{s}\fi{}ek}},\ }\href
  {https://doi.org/10.1103/PhysRevE.87.012118} {\bibfield  {journal} {\bibinfo
  {journal} {Phys. Rev. E}\ }\textbf {\bibinfo {volume} {87}},\ \bibinfo
  {pages} {012118} (\bibinfo {year} {2013})}\BibitemShut {NoStop}%
\bibitem [{\citenamefont {Serbyn}\ \emph {et~al.}(2021)\citenamefont {Serbyn},
  \citenamefont {Abanin},\ and\ \citenamefont {Papi{\'c}}}]{Serbyn2021}%
  \BibitemOpen
  \bibfield  {author} {\bibinfo {author} {\bibfnamefont {M.}~\bibnamefont
  {Serbyn}}, \bibinfo {author} {\bibfnamefont {D.~A.}\ \bibnamefont {Abanin}},\
  and\ \bibinfo {author} {\bibfnamefont {Z.}~\bibnamefont {Papi{\'c}}},\ }\href
  {https://doi.org/10.1038/s41567-021-01230-2} {\bibfield  {journal} {\bibinfo
  {journal} {Nat. Phys.}\ }\textbf {\bibinfo {volume} {17}},\ \bibinfo {pages}
  {675} (\bibinfo {year} {2021})}\BibitemShut {NoStop}%
\bibitem [{\citenamefont {Moudgalya}\ \emph {et~al.}()\citenamefont
  {Moudgalya}, \citenamefont {Bernevig},\ and\ \citenamefont
  {Regnault}}]{Moudgalya2021}%
  \BibitemOpen
  \bibfield  {author} {\bibinfo {author} {\bibfnamefont {S.}~\bibnamefont
  {Moudgalya}}, \bibinfo {author} {\bibfnamefont {B.~A.}\ \bibnamefont
  {Bernevig}},\ and\ \bibinfo {author} {\bibfnamefont {N.}~\bibnamefont
  {Regnault}},\ }\href {https://doi.org/10.48550/arXiv.2109.00548} {\bibinfo
  {journal} {arXiv:2109.00548}\ }\BibitemShut {NoStop}%
\bibitem [{\citenamefont {Papi{\'c}}()}]{Papic2021}%
  \BibitemOpen
\bibfield  {journal} {  }\bibfield  {author} {\bibinfo {author} {\bibfnamefont
  {Z.}~\bibnamefont {Papi{\'c}}},\ }\href
  {https://doi.org/10.48550/arXiv.2108.03460} {\bibinfo  {journal}
  {arXiv:2108.03460}\ }\BibitemShut {NoStop}%
\bibitem [{\citenamefont {Nandkishore}\ and\ \citenamefont
  {Huse}(2015)}]{Nandkishore2015}%
  \BibitemOpen
\bibfield  {journal} {  }\bibfield  {author} {\bibinfo {author} {\bibfnamefont
  {R.}~\bibnamefont {Nandkishore}}\ and\ \bibinfo {author} {\bibfnamefont
  {D.~A.}\ \bibnamefont {Huse}},\ }\href
  {https://doi.org/10.1146/annurev-conmatphys-031214-014726} {\bibfield
  {journal} {\bibinfo  {journal} {Annu. Rev. Condens. Matter Phys.}\ }\textbf
  {\bibinfo {volume} {6}},\ \bibinfo {pages} {15} (\bibinfo {year}
  {2015})}\BibitemShut {NoStop}%
\bibitem [{\citenamefont {Abanin}\ and\ \citenamefont
  {Papić}(2017)}]{Abanin2017}%
  \BibitemOpen
  \bibfield  {author} {\bibinfo {author} {\bibfnamefont {D.~A.}\ \bibnamefont
  {Abanin}}\ and\ \bibinfo {author} {\bibfnamefont {Z.}~\bibnamefont
  {Papić}},\ }\href {https://doi.org/https://doi.org/10.1002/andp.201700169}
  {\bibfield  {journal} {\bibinfo  {journal} {Ann. Phys. (Berlin)}\ }\textbf
  {\bibinfo {volume} {529}},\ \bibinfo {pages} {1700169} (\bibinfo {year}
  {2017})}\BibitemShut {NoStop}%
\bibitem [{\citenamefont {Alet}\ and\ \citenamefont
  {Laflorencie}(2018)}]{Alet2018}%
  \BibitemOpen
  \bibfield  {author} {\bibinfo {author} {\bibfnamefont {F.}~\bibnamefont
  {Alet}}\ and\ \bibinfo {author} {\bibfnamefont {N.}~\bibnamefont
  {Laflorencie}},\ }\href
  {https://doi.org/https://doi.org/10.1016/j.crhy.2018.03.003} {\bibfield
  {journal} {\bibinfo  {journal} {C. R. Phys.}\ }\textbf {\bibinfo {volume}
  {19}},\ \bibinfo {pages} {498} (\bibinfo {year} {2018})}\BibitemShut
  {NoStop}%
\bibitem [{\citenamefont {Abanin}\ \emph {et~al.}(2019)\citenamefont {Abanin},
  \citenamefont {Altman}, \citenamefont {Bloch},\ and\ \citenamefont
  {Serbyn}}]{Abanin2019}%
  \BibitemOpen
  \bibfield  {author} {\bibinfo {author} {\bibfnamefont {D.~A.}\ \bibnamefont
  {Abanin}}, \bibinfo {author} {\bibfnamefont {E.}~\bibnamefont {Altman}},
  \bibinfo {author} {\bibfnamefont {I.}~\bibnamefont {Bloch}},\ and\ \bibinfo
  {author} {\bibfnamefont {M.}~\bibnamefont {Serbyn}},\ }\href
  {https://doi.org/10.1103/RevModPhys.91.021001} {\bibfield  {journal}
  {\bibinfo  {journal} {Rev. Mod. Phys.}\ }\textbf {\bibinfo {volume} {91}},\
  \bibinfo {pages} {021001} (\bibinfo {year} {2019})}\BibitemShut {NoStop}%
\bibitem [{\citenamefont {Pal}\ and\ \citenamefont {Huse}(2010)}]{Pal2010}%
  \BibitemOpen
  \bibfield  {author} {\bibinfo {author} {\bibfnamefont {A.}~\bibnamefont
  {Pal}}\ and\ \bibinfo {author} {\bibfnamefont {D.~A.}\ \bibnamefont {Huse}},\
  }\href {https://doi.org/10.1103/PhysRevB.82.174411} {\bibfield  {journal}
  {\bibinfo  {journal} {Phys. Rev. B}\ }\textbf {\bibinfo {volume} {82}},\
  \bibinfo {pages} {174411} (\bibinfo {year} {2010})}\BibitemShut {NoStop}%
\bibitem [{\citenamefont {Luitz}\ \emph {et~al.}(2015)\citenamefont {Luitz},
  \citenamefont {Laflorencie},\ and\ \citenamefont {Alet}}]{Luitz2015}%
  \BibitemOpen
  \bibfield  {author} {\bibinfo {author} {\bibfnamefont {D.~J.}\ \bibnamefont
  {Luitz}}, \bibinfo {author} {\bibfnamefont {N.}~\bibnamefont {Laflorencie}},\
  and\ \bibinfo {author} {\bibfnamefont {F.}~\bibnamefont {Alet}},\ }\href
  {https://doi.org/10.1103/PhysRevB.91.081103} {\bibfield  {journal} {\bibinfo
  {journal} {Phys. Rev. B}\ }\textbf {\bibinfo {volume} {91}},\ \bibinfo
  {pages} {081103} (\bibinfo {year} {2015})}\BibitemShut {NoStop}%
\bibitem [{\citenamefont {Serbyn}\ \emph
  {et~al.}(2013{\natexlab{a}})\citenamefont {Serbyn}, \citenamefont
  {Papi\ifmmode~\acute{c}\else \'{c}\fi{}},\ and\ \citenamefont
  {Abanin}}]{Serbyn2013a}%
  \BibitemOpen
  \bibfield  {author} {\bibinfo {author} {\bibfnamefont {M.}~\bibnamefont
  {Serbyn}}, \bibinfo {author} {\bibfnamefont {Z.}~\bibnamefont
  {Papi\ifmmode~\acute{c}\else \'{c}\fi{}}},\ and\ \bibinfo {author}
  {\bibfnamefont {D.~A.}\ \bibnamefont {Abanin}},\ }\href
  {https://doi.org/10.1103/PhysRevLett.111.127201} {\bibfield  {journal}
  {\bibinfo  {journal} {Phys. Rev. Lett.}\ }\textbf {\bibinfo {volume} {111}},\
  \bibinfo {pages} {127201} (\bibinfo {year} {2013}{\natexlab{a}})}\BibitemShut
  {NoStop}%
\bibitem [{\citenamefont {\ifmmode \check{Z}\else
  \v{Z}\fi{}nidari\ifmmode~\check{c}\else \v{c}\fi{}}\ \emph
  {et~al.}(2008)\citenamefont {\ifmmode \check{Z}\else
  \v{Z}\fi{}nidari\ifmmode~\check{c}\else \v{c}\fi{}}, \citenamefont {Prosen},\
  and\ \citenamefont {Prelov\ifmmode~\check{s}\else
  \v{s}\fi{}ek}}]{Znidaric2008}%
  \BibitemOpen
  \bibfield  {author} {\bibinfo {author} {\bibfnamefont {M.}~\bibnamefont
  {\ifmmode \check{Z}\else \v{Z}\fi{}nidari\ifmmode~\check{c}\else
  \v{c}\fi{}}}, \bibinfo {author} {\bibfnamefont {T.}~\bibnamefont {Prosen}},\
  and\ \bibinfo {author} {\bibfnamefont {P.}~\bibnamefont
  {Prelov\ifmmode~\check{s}\else \v{s}\fi{}ek}},\ }\href
  {https://doi.org/10.1103/PhysRevB.77.064426} {\bibfield  {journal} {\bibinfo
  {journal} {Phys. Rev. B}\ }\textbf {\bibinfo {volume} {77}},\ \bibinfo
  {pages} {064426} (\bibinfo {year} {2008})}\BibitemShut {NoStop}%
\bibitem [{\citenamefont {Bardarson}\ \emph {et~al.}(2012)\citenamefont
  {Bardarson}, \citenamefont {Pollmann},\ and\ \citenamefont
  {Moore}}]{Bardarson2012}%
  \BibitemOpen
  \bibfield  {author} {\bibinfo {author} {\bibfnamefont {J.~H.}\ \bibnamefont
  {Bardarson}}, \bibinfo {author} {\bibfnamefont {F.}~\bibnamefont
  {Pollmann}},\ and\ \bibinfo {author} {\bibfnamefont {J.~E.}\ \bibnamefont
  {Moore}},\ }\href {https://doi.org/10.1103/PhysRevLett.109.017202} {\bibfield
   {journal} {\bibinfo  {journal} {Phys. Rev. Lett.}\ }\textbf {\bibinfo
  {volume} {109}},\ \bibinfo {pages} {017202} (\bibinfo {year}
  {2012})}\BibitemShut {NoStop}%
\bibitem [{\citenamefont {Serbyn}\ \emph
  {et~al.}(2013{\natexlab{b}})\citenamefont {Serbyn}, \citenamefont
  {Papi\ifmmode~\acute{c}\else \'{c}\fi{}},\ and\ \citenamefont
  {Abanin}}]{Serbyn2013b}%
  \BibitemOpen
  \bibfield  {author} {\bibinfo {author} {\bibfnamefont {M.}~\bibnamefont
  {Serbyn}}, \bibinfo {author} {\bibfnamefont {Z.}~\bibnamefont
  {Papi\ifmmode~\acute{c}\else \'{c}\fi{}}},\ and\ \bibinfo {author}
  {\bibfnamefont {D.~A.}\ \bibnamefont {Abanin}},\ }\href
  {https://doi.org/10.1103/PhysRevLett.110.260601} {\bibfield  {journal}
  {\bibinfo  {journal} {Phys. Rev. Lett.}\ }\textbf {\bibinfo {volume} {110}},\
  \bibinfo {pages} {260601} (\bibinfo {year} {2013}{\natexlab{b}})}\BibitemShut
  {NoStop}%
\bibitem [{\citenamefont {Huse}\ \emph {et~al.}(2014)\citenamefont {Huse},
  \citenamefont {Nandkishore},\ and\ \citenamefont {Oganesyan}}]{Huse2014}%
  \BibitemOpen
  \bibfield  {author} {\bibinfo {author} {\bibfnamefont {D.~A.}\ \bibnamefont
  {Huse}}, \bibinfo {author} {\bibfnamefont {R.}~\bibnamefont {Nandkishore}},\
  and\ \bibinfo {author} {\bibfnamefont {V.}~\bibnamefont {Oganesyan}},\ }\href
  {https://doi.org/10.1103/PhysRevB.90.174202} {\bibfield  {journal} {\bibinfo
  {journal} {Phys. Rev. B}\ }\textbf {\bibinfo {volume} {90}},\ \bibinfo
  {pages} {174202} (\bibinfo {year} {2014})}\BibitemShut {NoStop}%
\bibitem [{\citenamefont {Khemani}\ \emph {et~al.}(2016)\citenamefont
  {Khemani}, \citenamefont {Pollmann},\ and\ \citenamefont
  {Sondhi}}]{Khemani2016}%
  \BibitemOpen
  \bibfield  {author} {\bibinfo {author} {\bibfnamefont {V.}~\bibnamefont
  {Khemani}}, \bibinfo {author} {\bibfnamefont {F.}~\bibnamefont {Pollmann}},\
  and\ \bibinfo {author} {\bibfnamefont {S.~L.}\ \bibnamefont {Sondhi}},\
  }\href {https://doi.org/10.1103/PhysRevLett.116.247204} {\bibfield  {journal}
  {\bibinfo  {journal} {Phys. Rev. Lett.}\ }\textbf {\bibinfo {volume} {116}},\
  \bibinfo {pages} {247204} (\bibinfo {year} {2016})}\BibitemShut {NoStop}%
\bibitem [{\citenamefont {Yu}\ \emph {et~al.}(2017)\citenamefont {Yu},
  \citenamefont {Pekker},\ and\ \citenamefont {Clark}}]{Xiongjie2017}%
  \BibitemOpen
  \bibfield  {author} {\bibinfo {author} {\bibfnamefont {X.}~\bibnamefont
  {Yu}}, \bibinfo {author} {\bibfnamefont {D.}~\bibnamefont {Pekker}},\ and\
  \bibinfo {author} {\bibfnamefont {B.~K.}\ \bibnamefont {Clark}},\ }\href
  {https://doi.org/10.1103/PhysRevLett.118.017201} {\bibfield  {journal}
  {\bibinfo  {journal} {Phys. Rev. Lett.}\ }\textbf {\bibinfo {volume} {118}},\
  \bibinfo {pages} {017201} (\bibinfo {year} {2017})}\BibitemShut {NoStop}%
\bibitem [{\citenamefont {Sierant}\ \emph {et~al.}(2020)\citenamefont
  {Sierant}, \citenamefont {Lewenstein},\ and\ \citenamefont
  {Zakrzewski}}]{Sierant2020}%
  \BibitemOpen
  \bibfield  {author} {\bibinfo {author} {\bibfnamefont {P.}~\bibnamefont
  {Sierant}}, \bibinfo {author} {\bibfnamefont {M.}~\bibnamefont
  {Lewenstein}},\ and\ \bibinfo {author} {\bibfnamefont {J.}~\bibnamefont
  {Zakrzewski}},\ }\href {https://doi.org/10.1103/PhysRevLett.125.156601}
  {\bibfield  {journal} {\bibinfo  {journal} {Phys. Rev. Lett.}\ }\textbf
  {\bibinfo {volume} {125}},\ \bibinfo {pages} {156601} (\bibinfo {year}
  {2020})}\BibitemShut {NoStop}%
\bibitem [{\citenamefont {Van~Beeumen}\ \emph {et~al.}(2020)\citenamefont
  {Van~Beeumen}, \citenamefont {Kahanamoku-Meyer}, \citenamefont {Yao},\ and\
  \citenamefont {Yang}}]{Beeumen2020}%
  \BibitemOpen
  \bibfield  {author} {\bibinfo {author} {\bibfnamefont {R.}~\bibnamefont
  {Van~Beeumen}}, \bibinfo {author} {\bibfnamefont {G.~D.}\ \bibnamefont
  {Kahanamoku-Meyer}}, \bibinfo {author} {\bibfnamefont {N.~Y.}\ \bibnamefont
  {Yao}},\ and\ \bibinfo {author} {\bibfnamefont {C.}~\bibnamefont {Yang}},\
  }in\ \href {https://doi.org/10.1145/3368474.3368497} {\emph {\bibinfo
  {booktitle} {Proceedings of the International Conference on High Performance
  Computing in Asia-Pacific Region}}}\ (\bibinfo {year} {2020})\ p.\ \bibinfo
  {pages} {179–187}\BibitemShut {NoStop}%
\bibitem [{\citenamefont {De~Roeck}\ and\ \citenamefont
  {Huveneers}(2017)}]{Roeck2017}%
  \BibitemOpen
  \bibfield  {author} {\bibinfo {author} {\bibfnamefont {W.}~\bibnamefont
  {De~Roeck}}\ and\ \bibinfo {author} {\bibfnamefont {F.}~\bibnamefont
  {Huveneers}},\ }\href {https://doi.org/10.1103/PhysRevB.95.155129} {\bibfield
   {journal} {\bibinfo  {journal} {Phys. Rev. B}\ }\textbf {\bibinfo {volume}
  {95}},\ \bibinfo {pages} {155129} (\bibinfo {year} {2017})}\BibitemShut
  {NoStop}%
\bibitem [{\citenamefont {Luitz}\ \emph {et~al.}(2017)\citenamefont {Luitz},
  \citenamefont {Huveneers},\ and\ \citenamefont {De~Roeck}}]{Luitz2017}%
  \BibitemOpen
  \bibfield  {author} {\bibinfo {author} {\bibfnamefont {D.~J.}\ \bibnamefont
  {Luitz}}, \bibinfo {author} {\bibfnamefont {F.}~\bibnamefont {Huveneers}},\
  and\ \bibinfo {author} {\bibfnamefont {W.}~\bibnamefont {De~Roeck}},\ }\href
  {https://doi.org/10.1103/PhysRevLett.119.150602} {\bibfield  {journal}
  {\bibinfo  {journal} {Phys. Rev. Lett.}\ }\textbf {\bibinfo {volume} {119}},\
  \bibinfo {pages} {150602} (\bibinfo {year} {2017})}\BibitemShut {NoStop}%
\bibitem [{\citenamefont {Vosk}\ \emph {et~al.}(2015)\citenamefont {Vosk},
  \citenamefont {Huse},\ and\ \citenamefont {Altman}}]{Vosk2015}%
  \BibitemOpen
  \bibfield  {author} {\bibinfo {author} {\bibfnamefont {R.}~\bibnamefont
  {Vosk}}, \bibinfo {author} {\bibfnamefont {D.~A.}\ \bibnamefont {Huse}},\
  and\ \bibinfo {author} {\bibfnamefont {E.}~\bibnamefont {Altman}},\ }\href
  {https://doi.org/10.1103/PhysRevX.5.031032} {\bibfield  {journal} {\bibinfo
  {journal} {Phys. Rev. X}\ }\textbf {\bibinfo {volume} {5}},\ \bibinfo {pages}
  {031032} (\bibinfo {year} {2015})}\BibitemShut {NoStop}%
\bibitem [{\citenamefont {Potter}\ \emph {et~al.}(2015)\citenamefont {Potter},
  \citenamefont {Vasseur},\ and\ \citenamefont {Parameswaran}}]{Potter2015}%
  \BibitemOpen
  \bibfield  {author} {\bibinfo {author} {\bibfnamefont {A.~C.}\ \bibnamefont
  {Potter}}, \bibinfo {author} {\bibfnamefont {R.}~\bibnamefont {Vasseur}},\
  and\ \bibinfo {author} {\bibfnamefont {S.~A.}\ \bibnamefont {Parameswaran}},\
  }\href {https://doi.org/10.1103/PhysRevX.5.031033} {\bibfield  {journal}
  {\bibinfo  {journal} {Phys. Rev. X}\ }\textbf {\bibinfo {volume} {5}},\
  \bibinfo {pages} {031033} (\bibinfo {year} {2015})}\BibitemShut {NoStop}%
\bibitem [{\citenamefont {Dumitrescu}\ \emph {et~al.}(2017)\citenamefont
  {Dumitrescu}, \citenamefont {Vasseur},\ and\ \citenamefont
  {Potter}}]{Dumitrescu2017}%
  \BibitemOpen
  \bibfield  {author} {\bibinfo {author} {\bibfnamefont {P.~T.}\ \bibnamefont
  {Dumitrescu}}, \bibinfo {author} {\bibfnamefont {R.}~\bibnamefont
  {Vasseur}},\ and\ \bibinfo {author} {\bibfnamefont {A.~C.}\ \bibnamefont
  {Potter}},\ }\href {https://doi.org/10.1103/PhysRevLett.119.110604}
  {\bibfield  {journal} {\bibinfo  {journal} {Phys. Rev. Lett.}\ }\textbf
  {\bibinfo {volume} {119}},\ \bibinfo {pages} {110604} (\bibinfo {year}
  {2017})}\BibitemShut {NoStop}%
\bibitem [{\citenamefont {Thiery}\ \emph {et~al.}(2018)\citenamefont {Thiery},
  \citenamefont {Huveneers}, \citenamefont {M\"uller},\ and\ \citenamefont
  {De~Roeck}}]{Thiery2018}%
  \BibitemOpen
  \bibfield  {author} {\bibinfo {author} {\bibfnamefont {T.}~\bibnamefont
  {Thiery}}, \bibinfo {author} {\bibfnamefont {F.}~\bibnamefont {Huveneers}},
  \bibinfo {author} {\bibfnamefont {M.}~\bibnamefont {M\"uller}},\ and\
  \bibinfo {author} {\bibfnamefont {W.}~\bibnamefont {De~Roeck}},\ }\href
  {https://doi.org/10.1103/PhysRevLett.121.140601} {\bibfield  {journal}
  {\bibinfo  {journal} {Phys. Rev. Lett.}\ }\textbf {\bibinfo {volume} {121}},\
  \bibinfo {pages} {140601} (\bibinfo {year} {2018})}\BibitemShut {NoStop}%
\bibitem [{\citenamefont {Zhang}\ \emph {et~al.}(2016)\citenamefont {Zhang},
  \citenamefont {Zhao}, \citenamefont {Devakul},\ and\ \citenamefont
  {Huse}}]{Zhang2016}%
  \BibitemOpen
  \bibfield  {author} {\bibinfo {author} {\bibfnamefont {L.}~\bibnamefont
  {Zhang}}, \bibinfo {author} {\bibfnamefont {B.}~\bibnamefont {Zhao}},
  \bibinfo {author} {\bibfnamefont {T.}~\bibnamefont {Devakul}},\ and\ \bibinfo
  {author} {\bibfnamefont {D.~A.}\ \bibnamefont {Huse}},\ }\href
  {https://doi.org/10.1103/PhysRevB.93.224201} {\bibfield  {journal} {\bibinfo
  {journal} {Phys. Rev. B}\ }\textbf {\bibinfo {volume} {93}},\ \bibinfo
  {pages} {224201} (\bibinfo {year} {2016})}\BibitemShut {NoStop}%
\bibitem [{\citenamefont {Goremykina}\ \emph {et~al.}(2019)\citenamefont
  {Goremykina}, \citenamefont {Vasseur},\ and\ \citenamefont
  {Serbyn}}]{Goremykina2019}%
  \BibitemOpen
  \bibfield  {author} {\bibinfo {author} {\bibfnamefont {A.}~\bibnamefont
  {Goremykina}}, \bibinfo {author} {\bibfnamefont {R.}~\bibnamefont
  {Vasseur}},\ and\ \bibinfo {author} {\bibfnamefont {M.}~\bibnamefont
  {Serbyn}},\ }\href {https://doi.org/10.1103/PhysRevLett.122.040601}
  {\bibfield  {journal} {\bibinfo  {journal} {Phys. Rev. Lett.}\ }\textbf
  {\bibinfo {volume} {122}},\ \bibinfo {pages} {040601} (\bibinfo {year}
  {2019})}\BibitemShut {NoStop}%
\bibitem [{\citenamefont {Morningstar}\ and\ \citenamefont
  {Huse}(2019)}]{Morningstar2019}%
  \BibitemOpen
  \bibfield  {author} {\bibinfo {author} {\bibfnamefont {A.}~\bibnamefont
  {Morningstar}}\ and\ \bibinfo {author} {\bibfnamefont {D.~A.}\ \bibnamefont
  {Huse}},\ }\href {https://doi.org/10.1103/PhysRevB.99.224205} {\bibfield
  {journal} {\bibinfo  {journal} {Phys. Rev. B}\ }\textbf {\bibinfo {volume}
  {99}},\ \bibinfo {pages} {224205} (\bibinfo {year} {2019})}\BibitemShut
  {NoStop}%
\bibitem [{\citenamefont {Morningstar}\ \emph {et~al.}(2020)\citenamefont
  {Morningstar}, \citenamefont {Huse},\ and\ \citenamefont
  {Imbrie}}]{Morningstar2020}%
  \BibitemOpen
  \bibfield  {author} {\bibinfo {author} {\bibfnamefont {A.}~\bibnamefont
  {Morningstar}}, \bibinfo {author} {\bibfnamefont {D.~A.}\ \bibnamefont
  {Huse}},\ and\ \bibinfo {author} {\bibfnamefont {J.~Z.}\ \bibnamefont
  {Imbrie}},\ }\href {https://doi.org/10.1103/PhysRevB.102.125134} {\bibfield
  {journal} {\bibinfo  {journal} {Phys. Rev. B}\ }\textbf {\bibinfo {volume}
  {102}},\ \bibinfo {pages} {125134} (\bibinfo {year} {2020})}\BibitemShut
  {NoStop}%
\bibitem [{\citenamefont {Dumitrescu}\ \emph {et~al.}(2019)\citenamefont
  {Dumitrescu}, \citenamefont {Goremykina}, \citenamefont {Parameswaran},
  \citenamefont {Serbyn},\ and\ \citenamefont {Vasseur}}]{Dumitrescu2019}%
  \BibitemOpen
  \bibfield  {author} {\bibinfo {author} {\bibfnamefont {P.~T.}\ \bibnamefont
  {Dumitrescu}}, \bibinfo {author} {\bibfnamefont {A.}~\bibnamefont
  {Goremykina}}, \bibinfo {author} {\bibfnamefont {S.~A.}\ \bibnamefont
  {Parameswaran}}, \bibinfo {author} {\bibfnamefont {M.}~\bibnamefont
  {Serbyn}},\ and\ \bibinfo {author} {\bibfnamefont {R.}~\bibnamefont
  {Vasseur}},\ }\href {https://doi.org/10.1103/PhysRevB.99.094205} {\bibfield
  {journal} {\bibinfo  {journal} {Phys. Rev. B}\ }\textbf {\bibinfo {volume}
  {99}},\ \bibinfo {pages} {094205} (\bibinfo {year} {2019})}\BibitemShut
  {NoStop}%
\bibitem [{\citenamefont {Ericsson}\ and\ \citenamefont
  {Ruhe}(1980)}]{Ericsson1980}%
  \BibitemOpen
  \bibfield  {author} {\bibinfo {author} {\bibfnamefont {T.}~\bibnamefont
  {Ericsson}}\ and\ \bibinfo {author} {\bibfnamefont {A.}~\bibnamefont
  {Ruhe}},\ }\href {https://doi.org/10.2307/2006390} {\bibfield  {journal}
  {\bibinfo  {journal} {Math. Comp.}\ }\textbf {\bibinfo {volume} {35}},\
  \bibinfo {pages} {1251} (\bibinfo {year} {1980})}\BibitemShut {NoStop}%
\bibitem [{\citenamefont {Bai}\ \emph {et~al.}(2000)\citenamefont {Bai},
  \citenamefont {Demmel}, \citenamefont {Dongarra}, \citenamefont {Ruhe},\ and\
  \citenamefont {van~der Vorst}}]{TemplatesEigen}%
  \BibitemOpen
  \bibinfo {editor} {\bibfnamefont {Z.}~\bibnamefont {Bai}}, \bibinfo {editor}
  {\bibfnamefont {J.}~\bibnamefont {Demmel}}, \bibinfo {editor} {\bibfnamefont
  {J.}~\bibnamefont {Dongarra}}, \bibinfo {editor} {\bibfnamefont
  {A.}~\bibnamefont {Ruhe}},\ and\ \bibinfo {editor} {\bibfnamefont
  {H.}~\bibnamefont {van~der Vorst}},\ eds.,\ \href
  {https://doi.org/10.1137/1.9780898719581} {\emph {\bibinfo {title}
  {{Templates for the Solution of Algebraic Eigenvalue Problems: A Practical
  Guide}}}}\ (\bibinfo  {publisher} {SIAM},\ \bibinfo {address}
  {Philadelphia},\ \bibinfo {year} {2000})\BibitemShut {NoStop}%
\bibitem [{\citenamefont {Schenk}\ \emph {et~al.}(2008)\citenamefont {Schenk},
  \citenamefont {Bollhöfer},\ and\ \citenamefont {Römer}}]{Schenk2008}%
  \BibitemOpen
  \bibfield  {author} {\bibinfo {author} {\bibfnamefont {O.}~\bibnamefont
  {Schenk}}, \bibinfo {author} {\bibfnamefont {M.}~\bibnamefont {Bollhöfer}},\
  and\ \bibinfo {author} {\bibfnamefont {R.~A.}\ \bibnamefont {Römer}},\
  }\href {https://doi.org/10.1137/070707002} {\bibfield  {journal} {\bibinfo
  {journal} {SIAM Review}\ }\textbf {\bibinfo {volume} {50}},\ \bibinfo {pages}
  {91} (\bibinfo {year} {2008})}\BibitemShut {NoStop}%
\bibitem [{\citenamefont {Paige}\ and\ \citenamefont
  {Saunders}(1975)}]{Paige1975}%
  \BibitemOpen
  \bibfield  {author} {\bibinfo {author} {\bibfnamefont {C.~C.}\ \bibnamefont
  {Paige}}\ and\ \bibinfo {author} {\bibfnamefont {M.~A.}\ \bibnamefont
  {Saunders}},\ }\href {https://doi.org/10.1137/0712047} {\bibfield  {journal}
  {\bibinfo  {journal} {SIAM J. Numer. Anal.}\ }\textbf {\bibinfo {volume}
  {12}},\ \bibinfo {pages} {617} (\bibinfo {year} {1975})}\BibitemShut
  {NoStop}%
\bibitem [{\citenamefont {Van~Loan}\ and\ \citenamefont
  {Golub}(2013)}]{MatrixComputations}%
  \BibitemOpen
  \bibfield  {author} {\bibinfo {author} {\bibfnamefont {C.~F.}\ \bibnamefont
  {Van~Loan}}\ and\ \bibinfo {author} {\bibfnamefont {G.~H.}\ \bibnamefont
  {Golub}},\ }\href@noop {} {\emph {\bibinfo {title} {{Matrix
  Computations}}}},\ \bibinfo {edition} {4th}\ ed.\ (\bibinfo  {publisher}
  {Johns Hopkins University Press},\ \bibinfo {address} {Baltimore},\ \bibinfo
  {year} {2013})\BibitemShut {NoStop}%
\bibitem [{\citenamefont {Pietracaprina}\ \emph {et~al.}(2018)\citenamefont
  {Pietracaprina}, \citenamefont {Macé}, \citenamefont {Luitz},\ and\
  \citenamefont {Alet}}]{Pietracaprina2018}%
  \BibitemOpen
  \bibfield  {author} {\bibinfo {author} {\bibfnamefont {F.}~\bibnamefont
  {Pietracaprina}}, \bibinfo {author} {\bibfnamefont {N.}~\bibnamefont
  {Macé}}, \bibinfo {author} {\bibfnamefont {D.~J.}\ \bibnamefont {Luitz}},\
  and\ \bibinfo {author} {\bibfnamefont {F.}~\bibnamefont {Alet}},\ }\href
  {https://doi.org/10.21468/SciPostPhys.5.5.045} {\bibfield  {journal}
  {\bibinfo  {journal} {SciPost Phys.}\ }\textbf {\bibinfo {volume} {5}},\
  \bibinfo {pages} {45} (\bibinfo {year} {2018})}\BibitemShut {NoStop}%
\bibitem [{\citenamefont {Oganesyan}\ and\ \citenamefont
  {Huse}(2007)}]{Oganesyan2007}%
  \BibitemOpen
  \bibfield  {author} {\bibinfo {author} {\bibfnamefont {V.}~\bibnamefont
  {Oganesyan}}\ and\ \bibinfo {author} {\bibfnamefont {D.~A.}\ \bibnamefont
  {Huse}},\ }\href {https://doi.org/10.1103/PhysRevB.75.155111} {\bibfield
  {journal} {\bibinfo  {journal} {Phys. Rev. B}\ }\textbf {\bibinfo {volume}
  {75}},\ \bibinfo {pages} {155111} (\bibinfo {year} {2007})}\BibitemShut
  {NoStop}%
\bibitem [{\citenamefont {Lieb}\ \emph {et~al.}(1961)\citenamefont {Lieb},
  \citenamefont {Schultz},\ and\ \citenamefont {Mattis}}]{Lieb1961}%
  \BibitemOpen
  \bibfield  {author} {\bibinfo {author} {\bibfnamefont {E.}~\bibnamefont
  {Lieb}}, \bibinfo {author} {\bibfnamefont {T.}~\bibnamefont {Schultz}},\ and\
  \bibinfo {author} {\bibfnamefont {D.}~\bibnamefont {Mattis}},\ }\href
  {https://doi.org/https://doi.org/10.1016/0003-4916(61)90115-4} {\bibfield
  {journal} {\bibinfo  {journal} {Ann. Phys.}\ }\textbf {\bibinfo {volume}
  {16}},\ \bibinfo {pages} {407} (\bibinfo {year} {1961})}\BibitemShut
  {NoStop}%
\bibitem [{\citenamefont {Nakamura}\ and\ \citenamefont
  {Todo}(2002)}]{Nakamura2002}%
  \BibitemOpen
  \bibfield  {author} {\bibinfo {author} {\bibfnamefont {M.}~\bibnamefont
  {Nakamura}}\ and\ \bibinfo {author} {\bibfnamefont {S.}~\bibnamefont
  {Todo}},\ }\href {https://doi.org/10.1103/PhysRevLett.89.077204} {\bibfield
  {journal} {\bibinfo  {journal} {Phys. Rev. Lett.}\ }\textbf {\bibinfo
  {volume} {89}},\ \bibinfo {pages} {077204} (\bibinfo {year}
  {2002})}\BibitemShut {NoStop}%
\bibitem [{\citenamefont {Kawamura}\ \emph {et~al.}(2017)\citenamefont
  {Kawamura}, \citenamefont {Yoshimi}, \citenamefont {Misawa}, \citenamefont
  {Yamaji}, \citenamefont {Todo},\ and\ \citenamefont
  {Kawashima}}]{Kawamura2017}%
  \BibitemOpen
  \bibfield  {author} {\bibinfo {author} {\bibfnamefont {M.}~\bibnamefont
  {Kawamura}}, \bibinfo {author} {\bibfnamefont {K.}~\bibnamefont {Yoshimi}},
  \bibinfo {author} {\bibfnamefont {T.}~\bibnamefont {Misawa}}, \bibinfo
  {author} {\bibfnamefont {Y.}~\bibnamefont {Yamaji}}, \bibinfo {author}
  {\bibfnamefont {S.}~\bibnamefont {Todo}},\ and\ \bibinfo {author}
  {\bibfnamefont {N.}~\bibnamefont {Kawashima}},\ }\href
  {https://doi.org/https://doi.org/10.1016/j.cpc.2017.04.006} {\bibfield
  {journal} {\bibinfo  {journal} {Comput. Phys. Commun.}\ }\textbf {\bibinfo
  {volume} {217}},\ \bibinfo {pages} {180} (\bibinfo {year}
  {2017})}\BibitemShut {NoStop}%
\bibitem [{\citenamefont {Affleck}\ and\ \citenamefont
  {Lieb}(1986)}]{Affleck1986}%
  \BibitemOpen
  \bibfield  {author} {\bibinfo {author} {\bibfnamefont {I.}~\bibnamefont
  {Affleck}}\ and\ \bibinfo {author} {\bibfnamefont {E.~H.}\ \bibnamefont
  {Lieb}},\ }\href {https://doi.org/https://doi.org/10.1007/BF00400304}
  {\bibfield  {journal} {\bibinfo  {journal} {Lett. Math. Phys.}\ }\textbf
  {\bibinfo {volume} {12}},\ \bibinfo {pages} {57} (\bibinfo {year}
  {1986})}\BibitemShut {NoStop}%
\bibitem [{CUD(2021)}]{CUDAProgrammingGuide}%
  \BibitemOpen
  \href@noop {} {\bibinfo {title} {{Programming Guide :: CUDA Toolkit
  Documentation}}},\ \bibinfo {howpublished}
  {\url{https://docs.nvidia.com/cuda/cuda-c-programming-guide/index.html}}
  (\bibinfo {year} {2021}),\ \bibinfo {note} {(Accessed on
  12/05/2021)}\BibitemShut {NoStop}%
\bibitem [{\citenamefont {Vidmar}\ and\ \citenamefont
  {Rigol}(2017)}]{Vidmar2017}%
  \BibitemOpen
  \bibfield  {author} {\bibinfo {author} {\bibfnamefont {L.}~\bibnamefont
  {Vidmar}}\ and\ \bibinfo {author} {\bibfnamefont {M.}~\bibnamefont {Rigol}},\
  }\href {https://doi.org/10.1103/PhysRevLett.119.220603} {\bibfield  {journal}
  {\bibinfo  {journal} {Phys. Rev. Lett.}\ }\textbf {\bibinfo {volume} {119}},\
  \bibinfo {pages} {220603} (\bibinfo {year} {2017})}\BibitemShut {NoStop}%
\bibitem [{\citenamefont {Huang}(2019)}]{Huang2019}%
  \BibitemOpen
  \bibfield  {author} {\bibinfo {author} {\bibfnamefont {Y.}~\bibnamefont
  {Huang}},\ }\href
  {https://doi.org/https://doi.org/10.1016/j.nuclphysb.2018.09.013} {\bibfield
  {journal} {\bibinfo  {journal} {Nucl. Phys. B}\ }\textbf {\bibinfo {volume}
  {938}},\ \bibinfo {pages} {594} (\bibinfo {year} {2019})}\BibitemShut
  {NoStop}%
\bibitem [{\citenamefont {Huang}(2021)}]{Huang2021}%
  \BibitemOpen
  \bibfield  {author} {\bibinfo {author} {\bibfnamefont {Y.}~\bibnamefont
  {Huang}},\ }\href
  {https://doi.org/https://doi.org/10.1016/j.nuclphysb.2021.115373} {\bibfield
  {journal} {\bibinfo  {journal} {Nucl. Phys. B}\ }\textbf {\bibinfo {volume}
  {966}},\ \bibinfo {pages} {115373} (\bibinfo {year} {2021})}\BibitemShut
  {NoStop}%
\bibitem [{\citenamefont {Atas}\ \emph {et~al.}(2013)\citenamefont {Atas},
  \citenamefont {Bogomolny}, \citenamefont {Giraud},\ and\ \citenamefont
  {Roux}}]{Atas2013}%
  \BibitemOpen
  \bibfield  {author} {\bibinfo {author} {\bibfnamefont {Y.~Y.}\ \bibnamefont
  {Atas}}, \bibinfo {author} {\bibfnamefont {E.}~\bibnamefont {Bogomolny}},
  \bibinfo {author} {\bibfnamefont {O.}~\bibnamefont {Giraud}},\ and\ \bibinfo
  {author} {\bibfnamefont {G.}~\bibnamefont {Roux}},\ }\href
  {https://doi.org/10.1103/PhysRevLett.110.084101} {\bibfield  {journal}
  {\bibinfo  {journal} {Phys. Rev. Lett.}\ }\textbf {\bibinfo {volume} {110}},\
  \bibinfo {pages} {084101} (\bibinfo {year} {2013})}\BibitemShut {NoStop}%
\bibitem [{\citenamefont {Barrett}\ \emph {et~al.}(1994)\citenamefont
  {Barrett}, \citenamefont {Berry}, \citenamefont {Chan}, \citenamefont
  {Demmel}, \citenamefont {Donato}, \citenamefont {Dongarra}, \citenamefont
  {Eijkhout}, \citenamefont {Pozo}, \citenamefont {Romine},\ and\ \citenamefont
  {der Vorst}}]{TemplatesLinear}%
  \BibitemOpen
  \bibfield  {author} {\bibinfo {author} {\bibfnamefont {R.}~\bibnamefont
  {Barrett}}, \bibinfo {author} {\bibfnamefont {M.}~\bibnamefont {Berry}},
  \bibinfo {author} {\bibfnamefont {T.~F.}\ \bibnamefont {Chan}}, \bibinfo
  {author} {\bibfnamefont {J.}~\bibnamefont {Demmel}}, \bibinfo {author}
  {\bibfnamefont {J.}~\bibnamefont {Donato}}, \bibinfo {author} {\bibfnamefont
  {J.}~\bibnamefont {Dongarra}}, \bibinfo {author} {\bibfnamefont
  {V.}~\bibnamefont {Eijkhout}}, \bibinfo {author} {\bibfnamefont
  {R.}~\bibnamefont {Pozo}}, \bibinfo {author} {\bibfnamefont {C.}~\bibnamefont
  {Romine}},\ and\ \bibinfo {author} {\bibfnamefont {H.~V.}\ \bibnamefont {der
  Vorst}},\ }\href {https://epubs.siam.org/doi/book/10.1137/1.9781611971538}
  {\emph {\bibinfo {title} {{Templates for the Solution of Linear Systems:
  Building Blocks for Iterative Methods, 2nd Edition}}}}\ (\bibinfo
  {publisher} {SIAM},\ \bibinfo {address} {Philadelphia, PA},\ \bibinfo {year}
  {1994})\BibitemShut {NoStop}%
\end{thebibliography}%

\end{document}